\documentclass[pdfa,a4paper,UKenglish,cleveref,autoref,thm-restate]{lipics-v2021}

\usepackage[bibliography=common]{apxproof}
\usepackage{todonotes}
\usepackage[noend]{algpseudocode}
\usepackage{xspace}
\usepackage{thmtools}

\newtheorem{problem}{Problem}

\newtheorem{conj}[theorem]{Conjecture}
\newcommand{\vals}[1]{h_{#1}}
\newcommand{\suma}[1]{\Sigma #1}
\newcommand{\set}[1]{\{#1\}}
\renewcommand{\O}{\mathcal{O}}
\newcommand{\N}{\mathbb{N}}
\newcommand{\Npos}{\mathbb{N}_+}
\newcommand{\pspace}{\textsc{PSpace}\xspace}
\newcommand{\nl}{\textsc{NL}\xspace}
\newcommand{\np}{\textsc{NP}\xspace}
\newcommand{\expspace}{\textsc{ExpSpace}\xspace}
\newcommand{\tower}{\textsc{Tower}\xspace}
\newcommand{\ackermann}{\textsc{Ackermann}\xspace}
\newcommand{\hyperackermann}{\textsc{hyperAckermann}\xspace}
\newcommand{\tran}{\longrightarrow}
\newcommand{\push}[1]{$\text{\sc push}(#1)$}
\newcommand{\pop}[1]{$\text{\sc pop}(#1)$}
\newcommand{\inc}[1]{\add{#1}{1}}
\newcommand{\dec}[1]{\sub{#1}{1}}
\newcommand{\coreadd}[2]{#1 \,\, +\!\!= \, #2}
\newcommand{\coresub}[2]{#1 \,\, -\!\!= \, #2}
\newcommand{\add}[2]{$\coreadd{#1}{#2}$}
\newcommand{\sub}[2]{$\coresub{#1}{#2}$}
\algrenewcommand{\algorithmiccomment}[1]{\qquad$\rightarrow$ #1}
\newcommand{\goto}[2]{\textbf{goto} {\footnotesize #1} \textbf{or} {\footnotesize #2}}
\newcommand{\testz}[1]{\textbf{zero?}~$#1$}
\newcommand{\vr}[1]{\sf{#1}}

\newcommand{\PROG}[4]{
\begin{minipage}{#1\linewidth}
\medskip
\underline{\bf #2:}
\label{#3}
\begin{algorithmic}[1]
#4
\end{algorithmic}
\end{minipage}
}
\newcommand{\PROGnoname}[3]{
\begin{minipage}{#1\linewidth}
\medskip
\label{#2}
\begin{algorithmic}[1]
#3
\end{algorithmic}
\end{minipage}
}
\newcommand{\prog}[1]{\mathcal{#1}}
\newcommand{\successor}[1]{\widetilde #1}
\newcommand{\FF}[1]{\mathcal{F}_{#1}}
\newcommand{\F}[1]{\mathbf{F}_{#1}}
\newcommand{\triple}[6]{\text{\sc Triple}(#1, #2, #3, #4, #5, #6)}
\newcommand{\DTIME}[1]{\text{\sc DTime}(#1)}
\newcommand{\cval}[1]{\N^{{#1}}}

\algnewcommand{\LLoop}[1]{\State\algorithmicloop \ \ #1}
\renewcommand*{\appendixprelim}{}

\title{New Lower Bounds for Reachability in Vector Addition Systems}

\author{Wojciech Czerwi{\'{n}}ski}{University of Warsaw, Poland}{wczerwin@mimuw.edu.p}{https://orcid.org/0000-0002-6169-868X}{
}
\author{Isma\"{e}l Jecker}{University of Warsaw, Poland \and FEMTO-ST, CNRS, Univ. Franche-Comt\'{e}, France}{ismael.jecker@gmail.com}{}{
}
\author{S{\l}awomir Lasota}{University of Warsaw, Poland}{s.lasota@uw.edu.pl}{https://orcid.org/0000-0001-8674-4470}{Supported by the ERC grant INFSYS, agreement no. 950398 and by the NCN grant 2021/41/B/ST6/00535.}
\author{J\'{e}r\^{o}me Leroux}{LaBRI, CNRS, Univ. Bordeaux, France}{jerome.leroux@labri.fr}{}{Supported by the grant ANR-17-CE40-0028 of the French National Research Agency ANR (project BRAVAS).}
\author{{\L}ukasz Orlikowski}{University of Warsaw, Poland}{lo418363@students.mimuw.edu.pl}{}{
}

\authorrunning{W. Czerwi\'{n}ski, I. Jecker, S. Lasota, J. Leroux, and {\L}. Orlikowski}
\Copyright{Wojciech Czerwi{\'{n}}ski, Isma\"{e}l Jecker, S{\l}awomir Lasota, J\'{e}r\^{o}me Leroux, and {\L}ukasz Orlikowski}



\acknowledgements{We are grateful to Weijun Chen for his sharp observation that brought to light a technical error in the initial version of this paper.}
\funding{\textit{Wojciech Czerwi{\'{n}}ski, Isma\"{e}l Jecker and {\L}ukasz Orlikowski}:
Supported by the ERC grant INFSYS, agreement no. 950398.}

\EventEditors{Patricia Bouyer and Srikanth Srinivasan}
\EventNoEds{2}
\EventLongTitle{43rd IARCS Annual Conference on Foundations of Software Technology and Theoretical Computer Science (FSTTCS 2023)}
\EventShortTitle{FSTTCS 2023}
\EventAcronym{FSTTCS}
\EventYear{2023}
\EventDate{December 18--20, 2023}
\EventLocation{IIIT Hyderabad, Telangana, India}
\EventLogo{}
\SeriesVolume{284}
\ArticleNo{35}

\nolinenumbers

\begin{document}

\maketitle

\begin{abstract}
We investigate the dimension-parametric complexity of the reachability problem in
vector addition systems with states (VASS) and its extension with pushdown stack (pushdown VASS).
Up to now, the problem is known to be $\FF d$-hard for VASS of dimension $3d+2$
(the complexity class $\FF d$ corresponds to the $k$th level of the fast-growing hierarchy),
and no essentially better  bound is known for pushdown VASS.
We provide a new construction that improves the lower bound for VASS: $\FF d$-hardness in dimension $2d+3$.
Furthermore, building on our new insights we show a new lower bound for pushdown VASS: $\FF d$-hardness in dimension $\frac d 2 + 6$.
This dimension-parametric lower bound is strictly stronger than the upper bound for VASS,
which suggests that the (still unknown) complexity of the reachability problem in pushdown VASS 
is higher than in plain VASS (where it is Ackermann-complete).
\end{abstract}

\section{Introduction}

Petri nets, equivalently presentable as vector addition systems with states (VASS), 
are an established model of concurrency with widespread applications.
The central algorithmic problem for this model
is the \emph{reachability problem} which asks whether from a given initial configuration 
there exists a sequence of valid execution steps reaching a given final configuration.
For a long time the complexity of this problem remained
one of the hardest open questions in verification of concurrent systems. 
In 2019 Leroux and Schmitz made a significant breakthrough
by providing an Ackermannian upper bound~\cite{DBLP:conf/lics/LerouxS19}.
With respect to the hardness,
the exponential space lower bound, shown by Lipton already in 1976~\cite{Lipton76},
remained the only known for over 40 years
until a breakthrough non-elementary lower bound
by Czerwi{\'n}ski, Lasota, Lazic, Leroux and Mazowiecki in 2019~\cite{DBLP:conf/stoc/CzerwinskiLLLM19,jacm}.
Finally, a matching Ackermannian lower bound announced in 2021 independently
by two teams, namely Czerwi{\'n}ski and Orlikowski~\cite{DBLP:conf/focs/CzerwinskiO21}
and Leroux~\cite{DBLP:conf/focs/Leroux21},
established the exact complexity of the problem.

However, despite the fact that the exact complexity of
the reachability problem for VASS is settled,
there are still significant gaps in our understanding of the problem.
One such gap is the complexity of the reachability problem
\emph{parametrised by the dimension},
namely deciding the reachability problem
for $d$-dimensional VASS ($d$-VASS) for fixed $d \in \N$.
Currently, the exact complexity bounds are only known for dimensions one and two.
In these cases, the complexity depends on the representations of numbers
in the transitions, either unary or binary.
For binary VASS (where the numbers are represented in binary)
the reachability problem is known to be \np-complete
for $1$-VASS~\cite{DBLP:conf/concur/HaaseKOW09}
and \pspace-complete for $2$-VASS~\cite{BlondinFGHM15}.
For unary VASS the problem is \nl-complete for both $1$-VASS (folklore)
and $2$-VASS~\cite{EnglertLT16}.

Much less is known for higher dimensions,
and it is striking that even in the case of $3$-VASS we have a huge complexity gap.
The best complexity upper bound comes from the above mentioned work of 
Leroux and Schmitz~\cite{DBLP:conf/lics/LerouxS19},
where it is proved that the reachability problem for $(d-4)$-VASS is in $\FF d$
(here $\FF d$ denotes the $d$th level of the Grzegorczyk hierarchy
of complexity classes,
which corresponds to the fast growing function hierarchy).
In particular this shows that the reachability problem for $3$-VASS is in $\FF 7$
(recall that $\FF 3 = \tower$).

The recent Ackermann-hardness results provide lower bounds
for the reachability problem in fixed dimensions.
The result of Czerwi{\'n}ski and Orlikowski~\cite{DBLP:conf/focs/CzerwinskiO21}
yields $\FF d$-hardness for $6d$-VASS,
while the result of Leroux~\cite{DBLP:conf/focs/Leroux21}
establishes $\FF d$-hardness for $(4d+5)$-VASS.
Lasota improved upon these results
and showed $\FF d$-hardness of the problem for $(3d+2)$-VASS~\cite{Las22}.
In~\cite{DBLP:conf/lics/CzerwinskiO22}, additional
results
were obtained for specific
dimensions:
\pspace-hardness for unary $5$-VASS,
\expspace-hardness for binary $6$-VASS
and \tower-hardness for unary $8$-VASS.

To summarise, despite significant research efforts
there are still several natural problems related to the VASS reachability problem
that present significant complexity gaps:
\begin{description}
\item[$\textsf{Q}_1$:]
What is the complexity of the reachability problem for VASS of dimension $3$?
It is known to be  $\pspace$-hard and in $\FF 7$;
\item[$\textsf{Q}_2$:]
What is highest dimension for which the complexity of the reachability problem is elementary?
It is known to fall within the range of $2$ to $8$;
\item[$\textsf{Q}_3$:]
What is the smallest constant $C$ such that
the complexity of the reachability problem
for $d$-VASS is in $\FF {Cd+o(d)}$?
It is known to fall within the range of $\frac{1}{3}$ to $1$.
\end{description}
In this work, we focus on addressing Question \textbf{\texorpdfstring{\boldmath $\textsf{Q}_3$}{Q\_3}}.
We present new and improved lower bounds,
first in the standard setting of VASS,
and then in the setting of pushdown VASS (PVASS)
which extend the VASS model by incorporating a pushdown stack.

\subparagraph*{VASS reachability.}
Our first main result is a new complexity lower bound for the reachability problem which improves the gap:

\begin{restatable}{theorem}{vass}
\label{thm:vass}
The reachability problem for $(2d+3)$-VASS is $\FF d$-hard.
\end{restatable}

A preliminary version of this result was presented in~\cite{DBLP:journals/corr/abs-2104-12695}.
In this revised version, we aim to present the result
in a conceptually simple framework by using the notion of \emph{triples},
a formalism that was originally developed in~\cite{DBLP:conf/stoc/CzerwinskiLLLM19},
and was heavily used in~\cite{DBLP:conf/focs/CzerwinskiO21}~and~\cite{Las22}.
We view this contribution as an important step towards understanding
the complexity of the VASS reachability problem
parametrised by the dimension.

\subparagraph*{PVASS reachability.}
The decidability of the reachability problem for PVASS
has been an important open problem for over a decade~\cite{DBLP:conf/fsttcs/AtigG11}. Despite efforts of the community, it
remains unknown even for PVASS of dimension~$1$ ($1$-PVASS),
namely automata with one counter and one pushdown stack.

\begin{conj}
The reachability problem for PVASS is decidable.
\end{conj}

It is important to acknowledge that the PVASS setting is complex,
and few decidability results are known.
Some progress has been made in the study of the \emph{coverability}
problem, a variant of the reachability problem 
which asks whether from a given initial configuration there exists a sequence
of valid execution steps that reaches a configuration
\emph{greater} than the given final configuration.
Notably, the coverability problem has been shown to be decidable
for $1$-PVASS~\cite{DBLP:conf/icalp/LerouxST15} (and mentioned to be in \expspace).

Interestingly, despite the slow progress on determining upper bounds for PVASS,
there is limited knownledge about lower bounds as well.
To the best of our knowledge, the only lower bound
that is  not directly implied by the results on VASS
concerns (again) the coverability problem for $1$-PVASS,
which has been established as \pspace-hard~\cite{DBLP:journals/ipl/EnglertHLLLS21}.

As for our contribution,
our second main result is the first
complexity lower bound for the PVASS reachability problem 
that is not immediately inherited from VASS:

\begin{restatable}{theorem}{pvass}
\label{thm:pvass}
The reachability problem for $(\lfloor \frac d 2 \rfloor + 6)$-PVASS
is $\FF {d}$-hard.
\end{restatable}

Notably, Theorem~\ref{thm:pvass} implies that for sufficiently large $d$
the reachability problem for $d$-PVASS is harder than
the problem for $(d+1)$-VASS (which is
a subclass of $d$-PVASS as the pushdown stack can keep track of one VASS counter).
Indeed the problem for $d$-PVASS is $\FF {2d-12}$-hard by Theorem~\ref{thm:pvass},
while the problem for $(d+1)$-VASS is in $\FF {d+5}$ by~\cite{DBLP:conf/lics/LerouxS19}.

While our results indicate that the reachability problem for PVASS
is harder than for VASS,
some known results about PVASS hint that
even higher lower bounds might be proved:
In~\cite{DBLP:conf/csl/LerouxPS14} it was shown that PVASS
are able to weakly compute functions of
hyper-Ackermannian
growth rate.
Based on this observation we propose the following two conjectures:

\begin{conj}
There exists a fixed dimension $d \in \N$ such that
the reachability problem for $d$-PVASS is \ackermann-hard.
\end{conj}

\begin{conj}
The reachability problem for PVASS is \hyperackermann-hard.
\end{conj}

\section{Preliminaries}

\subparagraph*{Fast-growing hierarchy.} Let $\Npos = \N\setminus\{0\}$ be the set of positive integers.
We define the complexity classes $\FF i$ corresponding to the $i$th level in the Grzegorczyk 
Hierarchy~\cite[Sect.~2.3, 4.1]{Schmitz16toct}.
To this aim we choose to use the following family of functions $\F i : \Npos\to\Npos$, indexed by $i\in \N$:
\begin{align} \label{eq:F1}
\F 0(n) = n+2, \qquad \F {i+1} = \successor{{\F i}} \quad \text{where } \quad
\successor{{\,F\,}}(n) = F^{n-1}(1) = \underbrace{F \circ F  \circ \ldots \circ F}_{n - 1}(1).
\end{align}
In particular, $\F 1(n) = 2n-1$ and $\F i(1) = 1$ for all $i\in\Npos$. 
Using the functions $\F i$, we define  the complexity classes
$\FF i$, indexed by $i\in \Npos$, of problems solvable in deterministic time
$\F i \circ \F {i-1}^m (n)$ for some $m\in\N$:
\[
\FF i = \bigcup_{m\in\N} \DTIME{\F i \circ \F {i-1}^m (n)}.
\]
Intuitively speaking, the class $\FF i$ contains all problems solvable in time $\F i(n)$, and is closed under reductions
computable in time of lower order $\F {i-1}^m(n)$, for some fixed $m\in\Npos$.
In particular, $\FF 3 = \tower$
(problems solvable in a tower of exponentials of time or space
whose height is an elementary function of input size).
The classes $\FF i$ are robust with respect to
the choice of fast-growing function hierarchy (see \cite[Sect.4.1]{Schmitz16toct}).
For $i\geq 3$,  instead of deterministic time, one could equivalently take nondeterministic time, or space as all these definitions collapse.

\subsection{Counter programs}

A \emph{counter program} (or simply a \emph{program}) is a sequence of (line-numbered) commands, each of which is 
of one of the following types:
\begin{quote}
\begin{tabular}{l@{\qquad}l}
\inc{\vr{x}}                       & (increment the counter~$\vr{x}$ by one) \\
\dec{\vr{x}}                       & (decrement the counter~$\vr{x}$ by one) \\
\goto{$L$}{$L'$}                   & (nondeterministically jump to either line~$L$ or line~$L'$) \\
\testz{\vr{x}}                     & (zero test: continue if counter~$\vr{x}$ equals $0$)
\end{tabular}
\end{quote}
Counter programs \emph{with a pushdown stack}
(or simply \emph{programs with stack})
are enhanced versions of plain counter programs
that incorporate a stack
containing a word over a fixed stack alphabet $\mathbf{S}$.
The stack content is modified by using
two additional types of commands:
\begin{quote}
\begin{tabular}{l@{\qquad}l}
\push s                       & (push the symbol $s\in\vr S$ at the top of the stack) \\
\pop s                       & (pop the symbol $s\in\vr S$ if it is at the top of the stack) 
\end{tabular}
\end{quote}
The command \pop s fails if the stack is empty or if the top symbol
is different from $s$.
A \emph{configuration} of a counter program with pushdown stack consists, as expected,
of a valuation of its counters plus a stack content. 

Counters are only allowed to have nonnegative values.

\subparagraph*{Conventions.}
We are particularly interested in counter programs \emph{without zero tests}, i.e., ones that use no zero test command.
In the sequel, unless specified explicitly, counter programs are implicitly assumed to be \underline{without zero tests}.
Moreover, we use the syntactic sugar \textbf{loop}, which iterates
a sequence of command a nondeterministic number of times (see Example~\ref{ex:1}).
Finally, we write consecutive increments and decrements of different variables on a single line and we use the following shorthands:
\begin{quote}
\begin{tabular}{l@{\qquad}l}
\add{\vr x}{m}                      & (increment the counter~$\vr{x}$ by $m$) \\
\sub{\vr x}{m}                     & (decrement the counter~$\vr{x}$ by $m$) \\
$\vr x \tran \vr y$                   & (decrement the counter~$\vr x$ by one and increment the counter~$\vr y$ by one)\\
$\vr x \overset{m}{\tran} \vr y$      & (decrement the counter~$\vr x$ by $m$ and increment the counter~$\vr y$ by $m$)
\end{tabular}
\end{quote}

\begin{example} \label{ex:1}
As an illustration, consider three different presentations of the same the program with three counters $\vr C = \set{\vr x, \vr y, \vr z}$:
\medskip
\\
\begin{minipage}[t]{0.25\textwidth}
\strut\vspace*{-\baselineskip}

\PROGnoname{0.9}{ex}{
\State \goto 2 6
\State \dec{\vr{x}}  \label{l:testz.iter} \label{l:testz.iter2}
\State \inc{\vr{y}}                
\State \add{\vr{z}}{2}                               \label{l:testz.end}
\State \goto 1 1
\State \inc{\vr z}
}
\end{minipage}\begin{minipage}[t]{0.25\textwidth}
\strut\vspace*{-\baselineskip}

\PROGnoname{0.9}{loop}{
\Loop
\State \dec{\vr{x}}  \State \inc{\vr{y}}                
\State \add{\vr{z}}{2}                              \EndLoop
\State \inc{\vr z}
}
\end{minipage}\begin{minipage}[t]{0.4\textwidth}
\strut\vspace*{-\baselineskip}

\PROGnoname{0.9}{loop2}{
\LLoop{
$\vr x \tran \vr y$ \quad                
\add{\vr{z}}{2}}                              \State \inc{\vr z}
}
\end{minipage}
\pagebreak
\\
The program repeats the block of commands in lines \ref{l:testz.iter}--\ref{l:testz.end} 
some number of times chosen nondeterministically (possibly zero, although not infinite because $\vr{x}$ is decreasing,
and hence its initial value bounds the number of iterations) and then increments $\vr z$.

We emphasise that counters are only permitted to have nonnegative values.  
In the program above, that is why the decrement in line~\ref{l:testz.iter2} works also as a non-zero test.
\end{example}

\subparagraph*{Runs.}
Consider a program with counters $\vr X$.
By $\cval{\vr X}$
we denote the set of all valuations of counters.
Given an initial valuation of counters,
a \emph{run} (or \emph{execution}) of a counter program is a finite 
sequence of executions of commands.
A run which has successfully finished we call \emph{complete};
otherwise, the run is \emph{partial}.  Observe that, due to a decrement that would cause a counter to become negative, 
a partial run may fail to continue because it is blocked from further execution.  
Moreover, due to nondeterminism of \textbf{goto}, a program may have various runs
from the same initial valuation.

Two programs $\prog P, \prog Q$ may be \emph{composed} by concatenating them, written $\prog P\ \prog Q$.
We silently assume the appropriate re-numbering of lines referred to by \textbf{goto} command in $\prog Q$.

Consider a distinguished set of counters $\vr Z \subseteq \vr X$.
A run of $\prog P$ is \emph{$\vr Z$-zeroing} if it is complete and all counters from $\vr Z$ are zero at the end.
Given two counter valuations $r, r'\in \cval{\vr X}$
we say that the program $\vr Z$-\emph{computes} the valuation $r$ from $r'$
if $\prog P$ has exactly one $\vr Z$-zeroing run from $r'$
and the end configuration is $r$.
We also say that the program \emph{$\vr Z$-computes nothing} from $r'$
if $\prog P$ has no $\vr Z$-zeroing run from $r'$.
For instance,
in Example~\ref{ex:1}
the programs $\set{\vr x}$-compute,
from any valuation satisfying $\vr x = n \in \mathbb{N}$ and $\vr y = \vr z = 0$,
the valuation satisfying $\vr x = 0$ (trivially),
$\vr y = n$ and $\vr z = 2n+1$.

\subparagraph*{Maximal iteration.}
The proofs of this paper often focus on the number of iterations of the \textbf{loop} construct.
Consider a program $\prog P$ containing a \emph{flat} \textbf{loop}, i.e., a loop whose body consists 
of only increment or decrement commands, such that each counter appears in at most one of these
commands  (like the programs in Example~\ref{ex:1}).
We say that this loop is \emph{maximally iterated} in a given run of a $\prog P$ if
some counter that is decremented in its body is zero at the exit from the loop. 
In particular, a maximally iterated loop could not be iterated any further without violation of the 
nonnegativity constraint on the decremented counter.
For instance,
the loop in Example~\ref{ex:1} is maximally iterated by the $\{ \vr x \}$-zeroing runs.
Needless to say, maximal iteration needs not happen in general, for instance the program in Example~\ref{ex:1}
has multiple complete runs that do not admit this property.

\section{Main results and structure of the paper}
Counter programs (without zero test)
provide an equivalent presentation to the standard models of Petri nets and VASS.
The transformations between these different models are straightforward and the blowups are polynomial.
For instance, a program can be transformed into a VASS by taking one state for each line of the program,
and adding an appropriate transition corresponding to each counter update instruction.
Note that the dimension of the VASS obtained is equal to the number of counters of the program.
In this paper, we focus solely on counter programs,
and we prove that the following problem is $\FF d$-hard for every $d \geq 3$: 

\begin{problem}\label{prob:reachability_hard}
\
\begin{description}
\item[\upshape Input:]
A program $\prog P$ using $2d + 3$ counters.
\item[\upshape Question:] Is there a complete run of $\prog P$
that starts and ends with all counters equal to $0$?
\end{description}
\end{problem}
The equivalence between programs and VASS
then directly leads to our first main result, Theorem~\ref{thm:vass}.

For programs with a pushdown stack,
$\FF d$-hardness can be achieved with less counters.
We show that the following problem is $\FF d$-hard for every $d \geq 3$: 

\begin{problem}\label{prob:Creachability_hard}
\
\begin{description}
\item[\upshape Input:]
A program with stack $\prog Q$ using $\lfloor \frac d 2 \rfloor + 6$ counters.
\item[\upshape Question:] Is there a complete run of $\prog Q$
that starts and ends with all counters equal to $0$?
\end{description}
\end{problem}
Again, the equivalence between programs and VASS
yields our second main result, Theorem~\ref{thm:pvass}.

Let us now introduce the known $\FF d$-hard problem
that we will reduce to Problems~\ref{prob:reachability_hard} and~\ref{prob:Creachability_hard}.
Fortunately, we do not have to search too far for it:
counter programs with two counters \emph{and zero tests}
are Turing-complete~\cite{minsky_67}.
This implies that the reachability problem is undecidable in that setting.
However, similarly to Turing machines,
decidability is recovered by imposing limitations on the
executions,
such as bounding their length, the maximal counter size,
or the number of zero tests.
For our purposes the latter is the easiest to use,
thus, we present the problem that we will reduce from,
a variation of the ``$\F k$-bounded Minsky Machine Halting Problem''
proved to be $\FF d$-complete in~\cite[Section 2.3.2]{Schmitz16toct}:
\begin{problem}\label{prob:zeroreachability_hard}
\
\begin{description}
\item[\upshape Input:]
A program $\prog P$ with two zero-tested counters, and a bound $n \in \Npos$.
\item[\upshape Question:] Is there a complete run of $\prog P$
that starts and ends with all counters equal to $0$
and does exactly $\F d (n)$ zero tests?
\end{description}
\end{problem}
\medskip
The rest of this section is devoted to the presentation of the tools we use to reduce
Problem~\ref{prob:zeroreachability_hard} to Problems~\ref{prob:reachability_hard} and~\ref{prob:Creachability_hard}.
Following the structure of similar reductions presented in~\cite{DBLP:conf/focs/CzerwinskiO21}~and~\cite{Las22},
our reduction is divided into two main steps.

In the first step, we show how a program \emph{without zero test}
can simulate a bounded number of zero tests.
To achieve this we rely on the concept of \emph{triples},
which are specific counter valuations that allow
to eliminate the zero tests and instead verify whether a particular invariant still holds at the term of the run.
However, doing so requires an initial triple directly proportional to the number of zero tests we aim to simulate.
Since we intend to simulate $\F d(n)$ zero tests, which is a rather large number,
directly applying this approach would result in an excessively large program.

Thus, in the second step, we construct compact \emph{amplifiers}.
These amplifiers are small programs that compute functions of substantial magnitude (such as $\F d$)
while using a small number of counters (namely $2d + 4$, or $\lfloor \frac d 2 \rfloor + 4$ counters along with a stack).

We now define formally the notions required for these two steps.
This will allow us to state the main lemmas proved in this paper,
and use them to construct the reduction proving our main result.
The proofs of the lemmas can then be found in the following sections.

\subparagraph*{Triples.}
The concept of \emph{triples}
plays a central role in all the constructions presented within this paper.
Given a set of counters $\vr {X}$, three distinguished counters $\vr a, \vr b, \vr c \in \vr X$
and $A, B\in\mathbb{N}$, we denote by
$\triple A B {\vr a} {\vr b} {\vr c} {\vr X}$
the counter valuation satisfying
\begin{align} \label{eq:triple}
\vr a = A, \quad
\vr b = B, \quad
\vr c = A\cdot(4^B-1), \quad
\vr x = 0 \textup{ for every } \vr x \in \vr X \setminus \set{\vr a, \vr b, \vr c}.
\end{align}
Informally we sometimes call such valuation a \emph{$B$-triple},
or simply a \emph{triple} 
over the counters $\vr a, \vr b, \vr c$.
The interest of triples lies in their ability to establish invariants
that enable the detection of unwanted behaviours in counter programs.
For instance, in Section~\ref{sec:zero},
we show how to use triples to replace zero tests by proving the following lemma:
\begin{restatable}{lemma}{zero}
\label{lem:zero}
Let $\prog P$ be a program of size $m$ using two zero-tested counters.
There exists a program $\prog P'$ of size $\O(m)$ with six counters $\vr X$ such that
for every $B \in \Npos$ the two following conditions are equivalent:
\begin{itemize}
\item 
There exists a complete run of $\prog P$ that starts and ends with all counters equal to $0$ and
performs exactly $B$ zero tests;
\item There exists a complete run of $\prog P'$
that starts in some configuration
$\triple A {B} {\vr a} {\vr b} {\vr c} {\vr X}$ with $A \in \Npos$
and ends in a configuration where all the counters except $\vr a$ are zero.
\end{itemize}
\end{restatable}

\subparagraph*{Amplifiers.}
The key contribution of this paper 
consists of the construction of two families of programs that transform $B$-triples into $\F{d}(B)$-triples.
We formalise this type of programs through the notion of \emph{amplifiers}.
Let $F : \Npos \to \Npos$ be a monotone function satisfying $F(n)\geq n$ for every $n\in\Npos$.
Consider a program $\prog P$ using the set of counters $\vr X$,
out of which we distinguish six counters
$\vr a, \vr b, \vr c, \vr b', \vr c', \vr t\in \vr X$, 
and a subset of counters 
$\vr Z \subseteq \vr X\setminus \set{\vr a, \vr b, \vr c, \vr b', \vr t}$ that contains $\vr c'$.
The tuple $(\prog P,(\vr a, b, c),(\vr a, b', c'),t,Z)$
is called \emph{$F$-amplifier} if for all $A,B\in\Npos$:
\begin{itemize}
\item
$\prog P$ $\vr Z$-computes
$\triple {A\cdot 4^{B-F(B)}} {F(B)} {\vr a} {\vr b} {\vr c} {\vr X}$
from $\triple {A} {B} {\vr a} {\vr b'} {\vr c'} {\vr X}$
if $A$ is divisible by $4^{(F(B)-B)}$;
\item $\prog P$ $\vr Z$-computes
nothing
from $\triple {A} {B} {\vr a} {\vr b'} {\vr c'} {\vr X}$
if $A$ is not divisible by $4^{(F(B)-B)}$.
\end{itemize}
An amplifier transforms $B$-triples on its \emph{input counters} $\vr a, \vr b', \vr c'$ 
into $F(B)$-triples on its \emph{output counters} $\vr a, \vr b, \vr c$.
Remark that the counter $\vr a$ is involved in both input and output.
The counters in $\vr Z$, called \emph{end counters}, are intuitively speaking assumed
to be $0$-checked after the completion of a run of $\prog P$.
The auxiliary counter $\vr t$ does not play a direct role
apart from \emph{not} being an input counter, an output counter nor an end counter.
This will prove useful in our constructions.
We note that no condition is imposed on the runs that start
from a counter valuation that is not a triple on the input counters,
nor on the runs that are not $\vr Z$-zeroing.

In Section~\ref{sec:ampl} we construct of a family of $\F{d}$-amplifiers:

\begin{restatable}{lemma}{amp}\label{lem:amp}
For every $d \in\Npos$ there exists
an $\F d$-amplifier
of size $\O(d)$ that uses $2d+4$ counters out of which $d$ are end counters. 
\end{restatable}

Furthermore, in Section~\ref{sec:pampl} we extend the notion of amplifiers
to programs with stack, and show how
a stack
can replace three quarters of the counters used in the previous construction:

\begin{restatable}{lemma}{amppushdown}\label{lem:amppushdown}
For every $d\in\Npos$ there exists an $\F d$-amplifier
of size $\O(d)$ that uses a stack and $\lfloor \frac d 2 \rfloor + 4$ counters out of which $\lfloor \frac d 2 \rfloor$ are end counters.
\end{restatable}

\subparagraph*{Proof of the main theorems.}
While the proofs of our key lemmas are delegated to the appropriate sections,
we can already show how these lemmas yield a reduction from Problem~\ref{prob:zeroreachability_hard}
to Problems~\ref{prob:reachability_hard} and~\ref{prob:Creachability_hard}.
Let us consider an instance of Problem~\ref{prob:zeroreachability_hard}, that is,
a 2-counter program with zero tests $\prog P$ and an integer $n \in \Npos$.
We transform this instance into an instance $\prog P''$ of Problem~\ref{prob:reachability_hard}
and an instance $\prog Q''$ of Problem~\ref{prob:Creachability_hard}.
These two programs rely on the program $\prog P'$ given by Lemma~\ref{lem:zero},
and the $\F d$-amplifiers $\prog P_d$ and $\prog Q_d$ given by Lemmas~\ref{lem:amp} and~\ref{lem:amppushdown}.\footnote{Note that,
as explained in Section~\ref{sec:pampl},
when $d$ is odd calling the program $\prog Q_d$ requires non-deterministically pushing 
the content of $\vr c'$ onto the stack along with the content of $\vr b'$, which we omit here.}

\begin{minipage}[t]{0.45\textwidth}
\strut\vspace*{-\baselineskip}

\PROG{0.9}{Program $\prog P''$}{prog:P''}{
\State{\add{{\vr b'}}{n}}\label{l:triple1}
\LLoop{\inc{\vr a} \quad \add{{\vr c'}}{4^n-1}}\label{l:triple2}
\State $\prog P_d$
\State $\prog P'$
\LLoop{\dec{\vr a}}\label{l:decA}
}
\end{minipage}\begin{minipage}[t]{0.45\textwidth}
\strut\vspace*{-\baselineskip}

\PROG{0.9}{Program $\prog Q''$}{prog:Q''}{
\State{\push{\vr (b')^{n}}}
\LLoop{\inc{\vr a} \quad \add{{\vr c'}}{4^n-1}}
\State $\prog Q_d$
\State $\prog P'$
\LLoop{\dec{\vr a}}
}
\end{minipage}

\medskip

Due to size estimations of Lemmas~\ref{lem:zero}, \ref{lem:amp} and~\ref{lem:amppushdown},
the sizes of programs $\prog P''$ and $\prog Q''$ are linear in the size of $P$ plus $n+d$.
We now prove that this is a valid reduction from Problem~\ref{prob:zeroreachability_hard}
to Problem~\ref{prob:reachability_hard}.
The proof for Problem~\ref{prob:Creachability_hard} is analogous
since the programs $\prog P_d$ and $\prog Q_d$ have identical effect on triples. 

We need to show that $\prog P''$ has a complete run that starts and ends with all counters equal to $0$ if and only if
$\prog P$ has a complete run that starts and ends with all counters equal to $0$ and that does exactly $\F d(n)$ zero tests.
In order to prove it, let us consider the structure of a hypothetical run $\pi$ of $\prog P''$
that starts and ends with all counters having value zero.
Starting from the configuration where all the counters are zero, 
Lines~\ref{l:triple1}--\ref{l:triple2} generate an arbitrary $n$-triple:
Progressing through line~\ref{l:triple1} and performing $A \in \Npos$ iterations of line~\ref{l:triple2}
results in the configuration $\triple {A} {n} {\vr a} {\vr b'} {\vr c'} {\vr X}$.
Next, it is important to note that the subrun of $\pi$ involving $\prog P_d$ zeroes all the end counters,
since these counters remain unchanged after the invocation of $\prog P_d$
and they have value zero at the end of $\pi$.
Consequently, according to the definition of an amplifier, $\prog P_d$ transforms the $n$-triple $\triple {A} {n} {\vr a} {\vr b'} {\vr c'} {\vr X}$
into an $\F d(n)$-triple $\triple {A'} {\F d (n)} {\vr a} {\vr b} {\vr c} {\vr X}$.
From there, the subrun involving $\prog P'$ must end in a configuration
where every counter except $\vr a$ has value zero since the final line of $\prog P''$ can only decrement $\vr a$.
Therefore, the run $\pi$ exists if and only if there exists a run of $\prog P'$ that bridges the gap,
starting from an $\F d(n)$-triple $\triple {A'} {\F d (n)} {\vr a} {\vr b} {\vr c} {\vr X}$
and ending in a configuration where all the counters except $\vr a$ are $0$.
By Lemma~\ref{lem:zero} we know that such a run of $\prog P'$ exists if and only if
$\prog P$ has a complete run that starts and ends with all counters equal to $0$
and does exactly $\F d(n)$ zero tests,
which shows that our reduction is valid.

To conclude, let us remark that, as defined here, the program $\prog P''$ uses $2d + 7$ counters:
the call to $\prog P_d$ requires $2d + 4$ counters,
and while $\prog P'$ shares the output counters $\vr a,b$ and $\vr c$ of $\prog P_d$,
it uses three more counters.
We now argue that four counters can be saved, so that our program matches the definition of Problem~\ref{prob:reachability_hard}.
\begin{itemize}
\item The value of the input counter $\vr b'$ is never incremented in the program $\prog P_d$ we construct.\footnote{Remark that this is  \emph{not}
stated explicitly by Lemma~\ref{lem:amp}, but it is a trivial property of the corresponding construction presented in Section~\ref{sec:ampl}.}
Therefore, since in $\prog P''$ the call to $\prog P_d$ \emph{always} starts with the value ${\vr b'} = n$,
we can get rid of the counter $\vr b'$ by replacing the call to $\prog P_d$ with
$n$ consecutive copies of $\prog P_d$ in which each instruction decrementing $\vr b'$ is replaced with a jump to the next copy.
\item 
The second optimisation consists in reusing the counters of $\prog P_d$.
Since $\prog P_d$ is an amplifier, at the term of the run it is sufficient to check that the end counters are $0$ to ensure that \emph{all}
the counters except the output counters $\vr a$ and $\vr c$ are $0$.
Therefore, while the call to $\prog P'$ in $\prog P''$ needs to keep the $d$ end counters of $\prog P_d$ untouched,
there are still $d$ freely reusable counters (not counting $\vr a,b$ and $\vr c$ that are already reused),
and we can pick any three of these to use in $\prog P'$ instead of adding fresh ones.
\end{itemize}
For the program $\prog Q''$ it is not possible to save counters is that way,
but one of the extra counters of $\prog P'$ can be loaded on the stack,
which results in a program with $\lfloor \frac d 2 \rfloor + 6$ counters.

\newcommand{\zeroX}{$\prog Zero(\vr x)$}
\newcommand{\zeroY}{$\prog Zero(\vr y)$}
\newcommand{\zeroZ}{$\prog Zero(\vr z)$}

\section{Triples as a replacement for zero tests}\label{sec:zero}
The goal of this section is to prove Lemma~\ref{lem:zero}.
\medskip
\\
The six counters of program $\prog P'$ will consists in two counters $\vr x , y$ simulating the two counters of $\prog P$,
three counters $\vr a,b,c$ containing the initial triple,
and an auxiliary counter $\vr t$.
The idea behind the construction is that we will replace the zero tests with the two gadgets  \zeroX{} and \zeroY{}
defined as follows:

\begin{minipage}[t]{0.4\textwidth}
\strut\vspace*{-\baselineskip}

\PROG{0.9}{Program \zeroX{}}{prog:zero-testX}{
\LLoop{$\vr a \tran \vr t \quad \vr c \tran \vr t$}\label{l:doubleA1}
\LLoop{$\vr y \tran \vr x \quad \vr c \tran \vr t$}\label{l:doubleY1}
\LLoop{$\vr t \tran \vr a \quad \vr c \tran \vr a$}\label{l:doubleA2}
\LLoop{$\vr x \tran \vr y \quad \vr c \tran \vr a$}\label{l:doubleY2}
\State{\dec {\vr b}}
}
\end{minipage}\begin{minipage}[t]{0.4\textwidth}
\strut\vspace*{-\baselineskip}

\PROG{0.9}{Program \zeroY{}}{prog:zero-testY}{
\LLoop{$\vr a \tran \vr t \quad \vr c \tran \vr t$}
\LLoop{$\vr x \tran \vr y \quad \vr c \tran \vr t$}
\LLoop{$\vr t \tran \vr a \quad \vr c \tran \vr a$}
\LLoop{$\vr y \tran \vr x \quad \vr c \tran \vr a$}
\State{\dec {\vr b}}
}
\end{minipage}
\medskip
\\
The functioning of these two programs revolves around
the following invariant:
\begin{align}
&\textsc{Invariant:} 
& &(\vr a +  \vr x + \vr y  + t) \cdot 4^{\vr b} = \vr a +  \vr x + \vr y  + \vr t + \vr c  \ \textup{ and } \
\vr t = 0;\notag\\
&\textsc{Broken invariant:}
& &(\vr a +  \vr x + \vr y  + t) \cdot 4^{\vr b} < \vr a +  \vr x + \vr y  + \vr t + \vr c.\notag
\end{align}
Remark that the broken invariant is more specific
than the negation of the invariant.
We now present a technical lemma showing that \zeroX{} and \zeroY{}
accurately replace zero tests.

\begin{lemma}\label{claim:gadget}
Let $\vr z \in \{ \vr x,y\}$.
From each configuration of \zeroZ{} where $\vr b > 0$, $\vr z = 0$ and the invariant holds
there is a unique complete run that maintains the invariant and preserves the values of $\vr x$ and $\vr y$.
All other runs starting from this configuration,
as well as all runs starting with a broken invariant
or a holding invariant with a value of $\vr z$ greater than $0$,
end with a broken invariant.
\end{lemma}

\begin{proof}
We show the result for \zeroX{}: the proof can easily be transferred to \zeroY{} by
exchanging the roles of $\vr x$ and $\vr y$.
Let us start by observing that  \zeroX{} globally  decrements the value of $\vr b$ by $1$
and preserves the sum $\vr a +  \vr x + \vr y + t + c$ since every line preserves it.
Therefore, in order to maintain a holding invariant, \zeroX{}
needs to quadruple the value of the sum of counters $\vr a + x + y + t$.
Similarly, to repair a broken invariant \zeroX{}
needs to increase the value of $\vr a +  \vr x + \vr y  + t$ by more than quadrupling it.
We now study \zeroX{} in detail to show that the latter is impossible,
and that the former only occurs under specific conditions that imply the statement of the lemma.
We split our analysis in two:
\begin{description}
\item[Lines \ref{l:doubleA1}--\ref{l:doubleY1}:]
The loop on line~\ref{l:doubleA1}, respectively \ref{l:doubleY1},
increases $\vr a + x + y + t$ by at most $\vr a$, respectively~$\vr y$.
Therefore the value of $\vr a + x + y + t$ is at most doubled,
which occurs only if initially $\vr t = \vr x = 0$
and if then both loops are maximally iterated, resulting in $\vr a = \vr y = 0$.
\item[Lines \ref{l:doubleA2}--\ref{l:doubleY2}:]
The loop on line~\ref{l:doubleA2}, respectively \ref{l:doubleY2},
increases $\vr a + x + y + t$ by at most $\vr t$, respectively~$\vr x$.
Therefore the value of $\vr a + x + y + t$ is at most doubled,
which happens only if
$\vr a = \vr y = 0$ upon reaching line~\ref{l:doubleA2}
and if then both loops are maximally iterated, resulting in $\vr t = \vr x = 0$.
\end{description}
Combining both parts, we get that the value 
of $\vr a +  \vr x + \vr y  + t$ is at most quadrupled
through a run of \zeroX{},
which only occurs if initially $\vr t = \vr x = 0$ and all the loops are maximally iterated.
This immediately implies that \zeroX{} cannot repair a broken invariant.
Moreover, to maintain a holding invariant
it is necessary to actually quadruple this value,
therefore $\vr x$ needs to be $0$ at the start;
the maximal iteration of the loops will cause $\vr t$ to be $0$ at the end;
and the content of $\vr y$ will be completely moved to $\vr x$
by line~\ref{l:doubleY1}
and then back to $\vr y$ by line~\ref{l:doubleY2},
remaining unchanged as required by the statement.
\end{proof}

We now use \zeroX{} and \zeroY{} to transform
$\prog P$ into a program $\prog P'$ satisfying Lemma~\ref{lem:zero}.
Formally, the program $\prog P'$ is obtained
by applying the following modifications
to $\prog P$:
\begin{itemize}
\item We add an increment (resp. decrement) of $\vr a$
to each line of $\prog P$ featuring a decrement (resp. increment) of $\vr x$ or $\vr y$
so that every line preserves the value of $\vr a +  x  + y + t$;
\item We replace each zero tests ``\testz{\vr{x}}''
with a copy of the program \zeroX{},
and each zero test ``\testz{\vr{y}}''
with a copy of the program \zeroY{};
\end{itemize}
The first modification ensures that
all the lines of $\prog P'$
except the calls to \zeroX{} and \zeroY{}
maintain the invariant.
We observe that for every complete run $\pi'$ of $\prog P'$
that starts in a triple $\triple A {B} {\vr a} {\vr b} {\vr c} {\vr X}$ and ends in a configuration
where all counters except $\vr a$ are zero,
the invariant is satisfied both at the beginning
(by the definition of a triple) and at the end.
Therefore, according to Lemma~\ref{claim:gadget}
the invariant is never broken throughout $\pi'$,
indicating that the calls to \zeroX{} and \zeroY{} accurately simulate zero tests.
Consequently, $\pi'$ can be transformed into a matching run $\pi$ of $\prog P$
with same values of $\vr x$ and $\vr y$.
In particular, $\pi$ starts and ends with both counters equal to $0$.
Additionally, the value of $\vr b$ goes from $B$ to $0$ along $\pi'$.
Since this value is decremented by one by each call to \zeroX{} or \zeroY{},
$\pi'$ goes through exactly $B$ such calls,
which translates into $\pi$ performing exactly $B$ zero tests.

To conclude the proof of Lemma~\ref{lem:zero},
remark that we can also transform every run of $\prog P$ that starts and ends with both counters equal to $0$ and performs $B$ zero tests
into a matching run of $\prog P'$ starting from some triple $\triple A {B} {\vr a} {\vr b} {\vr c} {\vr X}$ and ending in a configuration
where all counters except $\vr a$ are zero
However, we must be cautious in choosing the initial value $A$ of $\vr a$ to be sufficiently high.
This ensures that we can increment $\vr x$ and $\vr y$ as high as required, despite the matching decrements of $\vr a$ added in $\prog P'$.

The size of $\prog P'$ is linear in the size of $\prog P$, as required. {}

\section{Amplifiers defined by counter programs}\label{sec:ampl}

This section is devoted to an inductive proof of Lemma~\ref{lem:amp}.
\medskip
\\
First we build  an $\F 1$-amplifier $\prog P_1$
with $6$ counters out of which one is an end counter (Lemma~\ref{lem:P1}).
Next we show how to
lift an arbitrary $F$-amplifier with $d$ counters
into an $\successor{{\,F\,}}$-amplifier by adding two counters
out of which one is an end counter (Lemma~\ref{lem:P}).
Applying $d-1$ times our lifting process to the program $\prog P_1$ 
yields an $\F d$-amplifier using $2d+4$ counters, proving Lemma~\ref{lem:amp}.

\subparagraph*{Strong amplifiers.}
For the purpose of induction step, namely for lifting $F$-amplifiers to $\successor{{\,F\,}}$-amplifiers,
we need a slight strengthening of the notion of amplifier.
An $F$-amplifier
$(\prog P,(\vr a, b, b),(\vr a, b', c'),t,Z)$
is called \emph{strong} if every run $\pi$ of $\prog P$ satisfies the following conditions
(let $\suma{\vr Z}$ stand for the sum of all counters in $\vr Z$):
\begin{enumerate}
\item The value of the sum  $\vr a + \vr c + \vr t + \suma{\vr Z}$ is the same at the start and at the end of $\pi$;
\item If $(\vr a + \vr c + \vr t + \suma{\vr Z} - \vr c') \cdot 4^{\vr b'} < \vr a + \vr c + \vr t + \suma{\vr Z}$
holds at the start of $\pi$, it also holds at the end.
\end{enumerate}

\subsection{Construction of the \texorpdfstring{\boldmath $\F 1$}{F\_\{\#1\} 1}-amplifier \texorpdfstring{\boldmath $\prog P_1$}{\#1 P\_1}}
\label{subsec:F1}

Consider the following program with 6 counters $\vr X = \set{\vr a, \vr b, \vr c, \vr b', \vr c', \vr t}$:

\PROG{0.9}{Program $\prog P_1$}{prog:P1}{
\LLoop{$\vr a \tran t$\label{l1:ini1}}
\label{l1:startini} 
\LLoop{$\vr t \tran a$  \quad \sub{\vr{c'}} 3 \quad
\add{\vr c} {3}\label{l1:dec1}}
\State \sub{\vr{b'}} 1 \quad \add{\vr b} {1}
\label{l1:endini} 
\Loop 
\label{l1:startit} 
\LLoop{$\vr a \tran t$ \quad
\sub{\vr{c'}}{1} \quad
\add{\vr t}{1}}\label{l1:main1}
\LLoop{$\vr c \tran a$  \quad \sub{\vr c'} 1 \quad \add{\vr{a}} 1} \label{l1:dec2}\label{l1:main2}
\LLoop{$\vr a \tran c$} \quad  \sub{\vr c'} 1 \quad \add{\vr c} 1\label{l1:main3}
\LLoop{$\vr t \overset{8}{\tran} c$ \quad \sub{\vr c'} 8 \quad \add{\vr a} 1 \quad \add{\vr c} 7  \label{l1:dec3}}\label{l1:main4}
\State \sub{\vr b'} 1 \quad \add{\vr b} {2} 
\EndLoop
\label{l1:endit} 
}
\pagebreak
\\
The program consists of an \emph{initialisation} step lines \ref{l1:startini}--\ref{l1:endini},
and an \emph{iteration} step lines \ref{l1:startit}--\ref{l1:endit}.
The initialisation step and the iteration step
both decrement $\vr b'$ by one and preserve the right-hand side $\vr a + c + t + c'$
of the invariant as every line of $\prog P_1$ preserves this sum.
Hence, to maintain the invariant the sum $\vr a + c + t$ needs to be quadrupled, 
and to repair a broken invariant
the sum $\vr a + c + t$ needs to be increased by an even larger amount.
We show that in both steps the latter is impossible,
and the former only happens if all loops are maximally iterated,
which implies the modification of the counter $\vr a$ required by the statements.

\begin{lemma} \label{lem:P1}
The program
$(\prog P_1,\vr X,(a,b,c),(a,b',c'),\vr t,\{c'\})$
is a strong $\F 1$-amplifier.
\end{lemma}
As an $\F 1$-amplifier, $\prog P_1$ is expected to map
each input  $\triple {A} {B} {\vr a} {\vr b'} {\vr c'} {\vr X}$ such that $4^{\F 1(B)-B}$ divides $A$
to the output
$\triple {A \cdot 4^{B-\F 1(B)}} {\F 1(B)} {\vr a} {\vr b} {\vr c} {\vr X}$.
Since $\F 1(B) = 2 B - 1$,
transforming the initial value ${\vr b'} = B$
into the final value ${\vr b} = \F 1(B)$ is easy:
$\prog P_1$ first decrements $\vr b$ by $1$
and increments $\vr b$ by $1$ once (line~\ref{l1:endini}),
and then increments $\vr b$ by $2$ whenever it decrements $\vr b'$ by $1$ (line~\ref{l1:endit}).
It is more complicated to transform the initial value ${\vr a} = A$
into the final value ${\vr a} = A \cdot 4^{B-\F 1(B)}$:
we need to divide $\F 1(B) - B = B-1$ times the content of $\vr a$ by $4$.
We prove that $\prog P_1$ does so by studying the following invariant:
\begin{align}
&\textsc{Invariant:} 
& &(\vr a + c + t) \cdot 4^{\vr b'} = \vr a + c + t + c' \ 
\textup{ and } \
\vr t = 0;\notag\\
&\textsc{Broken invariant:}
& &(\vr a + c + t) \cdot 4^{\vr b'} < \vr a + c + t + c'.\notag
\end{align}
Notice that saying that the invariant
is broken is more specific than saying that the invariant does not hold.
We now present two technical lemmas
describing how the invariant evolves along both steps of $\prog P_1$.
\begin{restatable}{lemma}{Pfirst}\label{claim:P1}
From each configuration where $\vr b'>0$, $\vr c = 0$ and the invariant holds,
there is a unique run through 
the initialisation step that maintains
the invariant and preserves the value of $\vr a$.
All the other runs starting from this configuration,
as well as the runs starting with a broken invariant,
end with a broken invariant.
\end{restatable}
\begin{proof}The value of  $\vr a + \vr c + \vr t $ is at most increased by $3 \cdot (\vr a + \vr t)$
along the initialisation part, 
as line \ref{l1:ini1} preserves this sum and moves the content of $\vr a$ to $\vr t$,
then line \ref{l1:dec1} increases this sum by at most $3$ times the value of~$\vr t$.
Therefore $\vr a + \vr c + \vr t $ is at most quadrupled, which implies
that the initialisation step cannot repair a broken invariant.
Moreover, to maintain a holding invariant the program
$\prog P_1$ needs to quadruple this sum.
This happens if and only if initially $\vr b' > 0$, $\vr c = 0$ and $\vr c' \geq 3 \cdot (\vr a+t)$ (this last condition is implied by the invariant);
and if then both loops are maximally iterated.
Finally, remark that upon maximal iteration of the loops 
$\vr a$ and $\vr t$ keep their initial values, as required.
\end{proof}

\begin{restatable}{lemma}{Psecond}\label{claim:P2}
From each configuration where $\vr b'>0$, $\vr a$ is divisible by $4$ and the invariant holds,
there is a unique run through 
the iteration step that maintains
the invariant and divides the value of $\vr a$ by $4$.
All the other runs starting from this configuration,
as well as those starting with a broken invariant or a holding invariant with a value of $\vr a$ not divisible by $4$,
end with a broken invariant.
\end{restatable}
\begin{proof}We divide our analysis of the iteration step in two parts:
\begin{description}
\item[Lines \ref{l1:main1}--\ref{l1:main2}]
The loop on line~\ref{l1:main1}, respectively line~\ref{l1:main2},
increases $\vr a + \vr c + \vr t$ by at most the value of $\vr a$, respectively $\vr c$.
Therefore the value of $\vr a + \vr c + \vr t$ is at most doubled,
which occurs only if initially $\vr t = 0$ and then both loops are maximally iterated,
resulting in $\vr c = 0$.
\item[Lines \ref{l1:main3}--\ref{l1:main4}]
The loop on line~\ref{l1:main3}, respectively~\ref{l1:main4},
increases $\vr a + \vr c + \vr t$ by at most the value of $\vr a$, respectively $\vr t$.
Therefore the value of $\vr a + \vr c + \vr t$ is at most doubled,
which occurs only if  $\vr c = 0$ upon reaching line~\ref{l1:main3} and then both loops are maximally iterated,
resulting in $\vr t = 0$.
\end{description}
Combining the two parts, we get that the value of $\vr a + \vr c + \vr t$ is at most quadrupled by the iteration step.
This directly implies that it is not possible to repair a broken invariant.
Moreover, to maintain a holding invariant this sum needs to be quadrupled.
This happens if and only if at the start of the iteration step 
$\vr b' > 0$, $\vr a$ is divisible by $4$, $\vr t = 0$ and $\vr c' \geq 3 \cdot (\vr a+c)$ (note that the last two conditions are implied by the invariant);
and if then the four loops are maximally iterated.\footnote{The fact that $\vr a$ is divisible by four is what allows
line~\ref{l1:main4} to be maximally iterated: if it is not the case,
the run would eventually get stuck with a content of $\vr t$ smaller than eight but greater than zero.}
To conclude, remark that maximally iterating lines \ref{l1:main1} and \ref{l1:main4} results in dividing the value of $\vr a$ by four:
line \ref{l1:main1} transfers twice the initial value of $\vr a$ to $\vr c$,
one eighth of which is then transferred back to $\vr a$ by line~\ref{l1:main4}.
\end{proof}

We proceed with the proof of Lemma~\ref{lem:P1}.
Let $A,B \in \mathbb{N}$
and let $\pi$ be a $\{c'\}$-zeroing run of the program $\prog P_1$
starting from 
$\triple {A} {B} {\vr a} {\vr b'} {\vr c'} {\vr X}$.
Initially the counters satisfy:
\[
{\vr a} = A,
\quad {\vr b'} = B,
\quad {\vr c'}= A \cdot (4^B-1),
\quad \vr b = \vr c = \vr t = 0.
\]
This directly implies that the invariant holds at the beginning of $\pi$.
Let us analyse the values of the counters at the end of $\pi$.
We immediately get $\vr c'= 0$ since $\pi$ is $\{c'\}$-zeroing.
Note that this implies that the invariant cannot be broken
as the counters always hold non-negative integer values.
As a consequence, Lemmas~\ref{claim:P1} and~\ref{claim:P2} imply
that the invariant still holds at the end of $\pi$,
and that $\pi$ is the unique $\{c'\}$-zeroing run of of $\prog P_1$
starting from
$\triple {A} {B} {\vr a} {\vr b'} {\vr c'} {\vr X}$.
To conclude the proof, we now show that at the end of $\pi$
all the counters match
$\triple {{A} \cdot {4^{B - \F 1(B)}}} {\F 1(B)} {\vr a} {\vr b} {\vr c} {\vr X}$:
\[
{\vr a} = A \cdot {4^{B-\F 1(B)}},
\quad {\vr b} = \F 1(B),
\quad {\vr c}= A \cdot (4^{B}-4^{B-\F 1(B)}),
\quad \vr b' = \vr c' = \vr t = 0.
\]
First, the invariant directly yields $\vr t = 0$,
and also $\vr b' = 0$ by using the fact that $\vr c'=0$.
As $\vr b'$ starts with value $B$ and is decremented by one
along the initialisation step and each iteration step,
we get that $\pi$ visits the iteration step $B-1$ times.
In turn, this implies that the final value of $\vr b$ is $2B-1 = \F 1(B)$,
and by Lemmas~\ref{claim:P1} and~\ref{claim:P2} we also get that the final value of $\vr a$
is $\frac{A}{4^{B-1}} = A \cdot {4^{B-\F 1(B)}}$.
Combining this with the fact that the initial value
$A \cdot 4^{B}$ of the sum $\vr a+c+t+c'$ 
is preserved along $\pi$
finally yields the appropriate value for $\vr c$.

Note that the run $\pi$ exists if and only if $ {4^{\F 1(B)-B}}$ divides $A$.
Otherwise Lemma~\ref{claim:P2} implies that the invariant is broken before the end of the run.
This proves that $\prog P_1$ is an $\F 1$-amplifier.
The fact that $\prog P_1$ is a \emph{strong} $\F 1$-amplifier then directly follows from Lemmas~\ref{claim:P1} and~\ref{claim:P2}.

\subsection{Construction of the \texorpdfstring{\boldmath $\successor{{\,F\,}}$}{F \textasciitilde}-amplifier \texorpdfstring{\boldmath $\successor{{\, \prog P\,}}$}{\#1 P \textasciitilde} from an \texorpdfstring{\boldmath $F$}{F}-amplifier \texorpdfstring{\boldmath $\prog P$}{\#1 P}}
\label{subsec:lift}
\enlargethispage{1.5\baselineskip}
Let $(\prog P,\vr X,(a,b,c),(a,b',c'),t,Z)$ be a strong
$F$-amplifier for some function $F : \Npos \to \Npos$.
We construct a strong $\successor{{\,F\,}}$-amplifier $\successor{{\,\prog P\,}}$
out of $\prog P$.
The program $\successor{{\,\prog P\,}}$
uses the counters of $\prog P$
plus two fresh input counters $\vr{b''}$ and $\vr{c''}$,
and it shares the output counters of $\prog P$:

\PROG{0.9}{Program $\successor{{\,\prog P\,}}$}{prog:P}{
\LLoop{$\vr a \tran t$  \label{l:ini1n}}
\label{l:startinitn} 
\LLoop{$\vr t \tran a$  \quad \sub{\vr{c''}} 3 \quad \add{\vr c} {3}}\label{l:dec1n}
\State \sub{\vr{b''}} 1 \quad \add{\vr b} {1}\label{l:endinitn}    \Loop 
\label{l:startmainn} 
\LLoop{$\vr a \tran t$}\label{l:main3n} \label{l:startn}
\LLoop{$\vr t \tran a$ \quad \sub{\vr{c''}} 3 \quad \add{\vr a} {3}}\label{l:endn}
\LLoop{$\vr c \tran c'$ \quad \sub{\vr{c''}} 3 \quad \add{\vr c'} {3}}\label{l:begcn}\label{l:dec2n}\label{l:endcn}
\LLoop{$\vr b \tran b'$}\label{l:begbn}\label{l:endbn}
\State $\prog P$  \label{l:invoke}
\State \sub{\vr{b''}} 1  \label{l:dec}
\EndLoop
\label{l:endmainn} 
}
\medskip
\\
Similarly to $\prog P_1$ the program $\successor{{\,\prog P\,}}$
consists of
an \emph{initialisation} step
lines \ref{l:startinitn}--\ref{l:endinitn}
(differing from $\prog P_1$ only by renaming counters),
and an \emph{iteration} step lines \ref{l:startmainn}--\ref{l:endmainn}
(differing significantly from $\prog P_1$).\begin{restatable}{lemma}{P}\label{lem:P}
For every strong
$F$-amplifier
$(\prog P,\vr X,(a,b,c),(a,b',c'),t,Z)$,
the program
$(\successor{{\,\prog P\,}},\vr X \cup \{\vr{b''}, \vr{c''}\},(a,b,c),(a,b'',c''),Z \cup \{c''\})$
is a strong $\successor{{\, F\,}}$-amplifier.
\end{restatable}~\\
The proof of Lemma~\ref{lem:P} can be found in Appendix~\ref{appendix:P}.
Here, we provide an overview of the main intuition behind it.
As an $\successor{{\, F\,}}$-amplifier,
$\successor{{\,\prog P\,}}$ is expected to map 
each input
$\triple {A} {B} {\vr a} {\vr b''} {\vr c''} {\vr X \cup \{\vr{b''}, \vr{c''}\}}$
such that $4^{\successor{{\, F\,}}(B)-F(B)}$ divides $A$ to the output
$\triple {{A} \cdot 4^{F(B)-\successor{{\, F\,}}(B)}} {\successor{{\, F\,}}(B)} {\vr a} {\vr b} {\vr c} {\vr X \cup \{\vr{b''}, \vr{c''}\}}$.
The intended behaviour of $\successor{{\,\prog P\,}}$ is straightforward:
Since for all $n \in \Npos$
the value $\successor{{\,F\,}}(n)$
is obtained by applying the function $F$ to $1$ for $n-1$ consecutive times,
we expect $\successor{{\,\prog P\,}}$ to apply
the program $\prog P$ exactly $B-1$ times to transform a $1$-triple into an 
$\successor{{\,\prog F\,}}(B)$-triple.
To show that $\successor{{\,\prog P\,}}$ behaves as expected,
we study the following invariant:
\begin{align}
&\textsc{Invariant:} 
& &(\vr a + c + t) \cdot 4^{\vr b''} = \vr a + c + t + c'' \ 
\textup{ and } \
\vr t = 0;\notag\\
&\textsc{Broken invariant:}
& &(\vr a + c + t) \cdot 4^{\vr b''} < \vr a + c + t + c''.\notag
\end{align}
The starting configuration $\triple {A} {B} {\vr a} {\vr b''} {\vr c''} {\vr X \cup \{\vr{b''}, \vr{c''}\}}$
satisfies the invariant.
The program $\successor{{\,\prog P\,}}$ is designed such that every run $\pi$
starting from such a configuration then satisfies:
\begin{itemize}
\item If along $\pi$ all the loops are maximally iterated
and all the calls to $\prog P$ are $Z$-zeroing,
then the invariant holds until the end of $\pi$.
Moreover, $\pi$ then matches the expected behaviour of $\successor{{\,\prog P\,}}$ described above.
In particular, $\pi$ will correctly amplify $B-1$ times via $\prog P$ a triple $\triple {{A}} {1} {\vr a} {\vr b} {\vr c} {\vr X}$,
thus ending in $\triple {{A} \cdot 4^{F(B)-\successor{{\, F\,}}(B)}} {\successor{{\, F\,}}(B)} {\vr a} {\vr b} {\vr c} {\vr X \cup \{\vr{b''}, \vr{c''}\}}$.
\item However, if $\pi$ fails to maximally iterate one loop,
or does a call to $\prog P$ that is not $Z$-zeroing,
then the invariant is irremediably broken, which implies that $\pi$ is not $(Z \cup \{c''\})$-zeroing.
\end{itemize}
This proves that $\successor{{\,\prog P\,}}$ is a strong $\successor{{\, F\,}}$-amplifier.

\begin{toappendix}
\section{Proof of Lemma~\ref{lem:amp}}\label{appendix:P}

\P*

The proof is structured as follows:
We begin with a technical lemma implying that
$\successor{{\,\prog P\,}}$ satisfies the two invariants required to be a
\emph{strong} amplifier (Claim~\ref{lem:strongAMP}).
Then, to show that $\successor{{\,\prog P\,}}$
is an $\successor{{\,\prog F\,}}$-amplifier,
we formalise the expected behaviour of
the runs of $\successor{{\,\prog P\,}}$ (Equations~\eqref{eq:w}--\eqref{eq:xB}),
we show that the runs that fit
this expected behaviour $(Z\cup\{c''\})$-compute $\successor{{\,\prog F\,}}$ (Claim~\ref{claim:goodAmp}),
and that the runs that do not fit this expected behaviour are
not $(Z\cup\{c''\})$-zeroing (Claim~\ref{claim:badAmp}).

\subparagraph*{Invariants of \texorpdfstring{\boldmath $\successor{{\,\prog P\,}}$}{\#1 P \textasciitilde}.}
Before delving into the intricate functioning of
$\successor{{\,\prog P\,}}$,
we show two invariants that hold for every run.
On top of being prerequisites for $\successor{{\,\prog P\,}}$
to qualify as a \emph{strong} amplifier,
these invariants offer valuable assistance in proving the next results
of this section.

\begin{claim}\label{lem:strongAMP}
The initialisation step and the iteration step of
$\successor{{\,\prog P\,}}$ both
preserve the value of
$\vr a + \vr c + \vr t + \suma{\vr Z} + \vr{c''}$
and either preserve or decrease the value of
$(\vr a + \vr c + \vr t + \suma{\vr Z}) \cdot 4^{\vr{b''}}$.
\end{claim}

\begin{claimproof}
We start by observing that
the sum $\vr a + \vr c + \vr t + \suma{\vr Z} + \vr{c''}$
stays constant along every run of $\successor{{\,\prog P\,}}$:
every command line preserves it,
including line~\ref{l:invoke} since $\prog P$ is a strong $F$-amplifier.
We now study the effect of
the initialisation and iteration steps
on the value of 
$(\vr a + \vr c + \vr t + \suma{\vr Z}) \cdot 4^{\vr{b''}}$.

\begin{description}
\item[Initialisation:]
The sum $\vr a + \vr c + \vr t + \suma{\vr Z}$
is preserved in lines~\ref{l:startinitn} and~\ref{l:endinitn}
and is increased in line~\ref{l:dec1n} by at most three times the value of $\vr t$,
thus it is at most quadrupled by the initialisation step.
Since $\vr b''$ is decremented by $1$ during the initialisation step,
this proves that the value of
$(\vr a + \vr c + \vr t + \suma{\vr Z}) \cdot 4^{\vr{b''}}$
is either preserved or decreased.
\item[Iteration:]
The sum $\vr a + \vr c + \vr t + \suma{\vr Z}$
is increased in lines \ref{l:startn}--\ref{l:endn}
by at most three times the value of $\vr a + \vr t$;
it is increased in line
\ref{l:begcn} by at most three times the value of $\vr c$;
and it is then preserved in lines~\ref{l:begbn},~\ref{l:invoke} and~\ref{l:dec}
(since $\prog P$ is a strong amplifier).
Hence the value of $\vr a + \vr c + \vr t + \suma{\vr Z}$
is at most quadrupled by each occurrence of the iteration step.
Since $\vr b''$ is decremented by $1$,
this shows that the value of
$(\vr a + \vr c + \vr t + \suma{\vr Z}) \cdot 4^{\vr{b''}}$
is either preserved or decreased. \qedhere
\end{description}
\end{claimproof}

\enlargethispage{-1\baselineskip}
\subparagraph*{Expected behaviour of \texorpdfstring{\boldmath $\successor{{\,\prog P\,}}$}{\#1 P \textasciitilde}.}
The intended behaviour of $\successor{{\,\prog P\,}}$ is straightforward.
We start with a $B$-triple over the counters $\vr a, b'', c''$.
The initialisation step establishes a $1$-triple over $\vr a, b,\vr c$.
Next, in the iteration step,
this $1$-triple is first moved to $\vr a, b',c'$,
and then $\prog P$ is invoked to transform it
into a $F(1)$-triple over $\vr a, b,\vr c$.
By repeating the iteration step $B-2$ more times,
we obtain a $F^{B-1}(1) = \successor{{\,\prog F\,}}(B)$-triple over $\vr a, b,\vr c$,
as expected from a $\successor{{\,\prog F\,}}$-amplifier.
We now formalise this expected behaviour as a set of equations.

Let $\pi$ be a run of $\successor{{\,\prog P\,}}$
that visits the iteration step $n$ times.
Let $w_0(\pi)$ denote the valuation of the counter set ${\vr X} \cup \{\vr b'',c''\}$
at the start of $\pi$,
and $x_n(\pi)$ denote the valuation of the counter set $\vr X$
at the end of $\pi$.
Moreover, for every $i = 0,1 \ldots, n-1$,
we use $x_i(\pi)$ and $y_{i}(\pi)$
to denote the valuation of $\vr X$
at the start of the $(i+1)$th iteration step of $\pi$
and at the start of the $(i+1)$th call to the program $\prog P$, respectively.
This notation allows us to formally express
the expected behaviour described earlier:
\begin{align}
w_0(\pi) = \ &
\triple {A} {B} {\vr a} {\vr b''} {\vr c''} {\vr X \cup \{b'',c''\}}
\label{eq:w}\\
x_0(\pi) = \ &
\triple {A} {1} {\vr a} {\vr b} {\vr c} {\vr X}
\label{eq:x0}\\
x_i(\pi) = \ &
\triple {A \cdot 4^{i+1-F^{i}(1)}} {F^{i}(1)} {\vr a} {\vr b} {\vr c} {\vr X}
\label{eq:x}\\
y_i(\pi) = \ &
\triple {A\cdot 4^{i+2-F^{i}(1)}} {F^{i}(1)} {\vr a} {\vr b'} {\vr c'} {\vr X}
\label{eq:y}\\
x_{B-1}(\pi) = \ &
\triple {A\cdot 4^{B-\successor{{\,F\,}}(B)}} {\successor{{\,F\,}}(B)} {\vr a} {\vr b} {\vr c} {\vr X}
\label{eq:xB}
\end{align}

\begin{figure}[b]
\centering
\newcolumntype{C}{>{\hfil$}p{2.0cm}<{$\hfil}}
$\begin{array}{r|C|C|C|C|C|C|C|}
\cline{2-6}
& \vr a & \vr b & \vr b' & \vr c  & \vr c'
\\
\cline{2-6}
w_0(\pi): & A & 0 & 0 & 0 & 0
\\
x_0(\pi): & A & 1 & 0 & A \cdot 2 - {\vr a} & 0
\\
x_i(\pi): & A \cdot  4^{i + 1 - F^{i}(1)} & F^i(1) & 0 & A \cdot 4^{i+1} - {\vr a} & 0
\\
y_i(\pi): & A \cdot  4^{i + 2 - F^{i}(1)} & 0 & F^i(1) & 0 & A \cdot 4^{i+2} - {\vr a}
\\
x_{B-1}(\pi): & A \cdot 4^{B - \successor{{\,F\,}}(B)} & \successor{{\,F\,}}(B) & 0 & A \cdot 4^{B} - {\vr a} & 0
\\
\cline{2-6}
\end{array}$
\caption{Individual counter values corresponding to the expressions
in Equations~\eqref{eq:w}--\eqref{eq:xB}.}
\label{fig:valueEvolutionAMP}
\end{figure}
~\\
The individual counter values corresponding to these equations
are listed in Figure~\ref{fig:valueEvolutionAMP}.

We split the set of runs of  $\successor{{\,\prog P\,}}$ in two parts: 
the \emph{good} runs, for which we show that Equations~\eqref{eq:w}--\eqref{eq:xB}
hold, and the \emph{bad} runs, that we prove to be 
non $(Z \cup \{c''\})$-zeroing.
Formally, we say that a run of $\successor{{\,\prog P\,}}$ is \emph{good} if
it goes through $B-1$ iteration steps;
if all the loops visited along it are maximally iterated;
and if all its calls to the program $\prog P$ are $Z$-zeroing.
By opposition, we describe as \emph{bad} the runs that fail to satisfy
at least one of these conditions.

\subparagraph*{Good runs.}
We prove that the good runs of $\successor{{\,\prog P\,}}$
compute the function $\successor{{\,F\,}}$:

\begin{claim}\label{claim:goodAmp}
Let $A,B \in \Npos$ be two positive integers.
Every good run of $\successor{{\,\prog P\,}}$
starting in $\triple A B {\vr a} {\vr b''} {\vr c''} {\vr X \cup \{b'',c''\}}$
satisfies Equations~\eqref{eq:w}--\eqref{eq:xB},
thus in particular it ends
in $\triple {A\cdot 4^{B-\successor{{\,F\,}}(B)}} {\successor{{\,F\,}}(B)} {\vr a} {\vr b} {\vr c} {\vr X \cup \{b'',c''\}}$.
Moreover, there exists such a run
if and only if $4^{\successor{{\,F\,}}(B)-B}$ divides $A$.
\end{claim}

\begin{claimproof}
Let $\pi$ be a good run of $\successor{{\,\prog P\,}}$ starting in 
$\triple A B {\vr a} {\vr b''} {\vr c''} {\vr X}$.
We immediately get that Equation~\eqref{eq:w} is satisfied.
To show that $\pi$ satisfies Equations~\eqref{eq:x0}--\eqref{eq:xB},
we prove via three inductive steps that Figure~\ref{fig:valueEvolutionAMP} is an accurate
depictions of the valuations $x_i(\pi)$ and $y_i(\pi)$
for every $i \in \{0,2,\ldots,B-1\}$:
\begin{enumerate}
\item\label{item:Init}
First, starting from a valuation satisfying $\vr c'' \geq \vr a$,
the effect of the initialisation step
with maximal iteration of the flat loops
is equivalent to the following sequence of assignments:
\[
\vr b'' := \vr b'' - 1,
\quad \quad \vr b := \vr b + 1,
\quad \quad \vr{c''} := \vr {c''} - \vr a,
\quad \quad \vr c := \vr a.
\]
This maps the first row of Figure~\ref{fig:valueEvolutionAMP}
to its second row,
thus Equation~\eqref{eq:x0} holds.
\item\label{item:Iter}
Next, starting from a valuation satisfying $\vr c'' \geq \vr a$,
the effect of lines~\ref{l:startn}--\ref{l:endbn} of $\successor{{\,\prog P\,}}$
with maximal iteration of the flat loops
is equivalent to the following sequence of assignments:
\[
\vr a := 2 \cdot \vr a,
\quad \quad \vr b' := \vr b,
\quad \quad \vr b := 0,
\quad \quad \vr{c''} := \vr {c''} - ({\vr a+c}),
\quad \quad \vr c' := \vr c' + 2 c,
\quad \quad \vr c := \vr 0.
\]
This maps the third row of Figure~\ref{fig:valueEvolutionAMP}
to its fourth row, thus whenever Equation~\eqref{eq:x} holds for some
$0 \in \{1,2,\ldots,B-2\}$,
so does Equation~\eqref{eq:y}.
\item\label{item:Call}
Finally, as  the program $\prog P$ is a strong $F$-amplifier,
for all $i \in \{0,1,\ldots,B-2\}$ it $Z$-computes
$\triple {A \cdot 4^{i+2-F^{i+1}(1)}} {F^{i+1}(1)} {\vr a} {\vr b} {\vr c} {\vr X}$
from
$\triple {A \cdot 4^{i+2-F^{i}(1)}} {F^{i}(1)} {\vr a} {\vr b'} {\vr c'} {\vr X}$.
Therefore, as every call to $\prog P$ along $\pi$ is $Z$-zeroing
since $\pi$ a good run, we get that if Equation~\eqref{eq:y} holds
for some $i \in \{0,1,\ldots,B-2\}$ then Equation~\eqref{eq:x} holds for $i+1$.
\end{enumerate}

To show that the run $\pi$ ends in
$\triple {A\cdot 4^{B-\successor{{\,F\,}}(B)}} {\successor{{\,F\,}}(B)} {\vr a} {\vr b} {\vr c} {\vr X \cup \{b'',c''\}}$,
we still need to address the values
of counters $b''$ and $c''$
(as Equation~\ref{eq:xB} only
specifies the value of the counter set $X$).
We directly get that the value of $b''$ is $0$ at the end of $\pi$:
$b''$ starts with value $B$ and is decremented once in the initialisation step
and in each of the $B-1$ occurrences of the iteration step.
Moreover, we also get that $c''$ is $0$ at the end of $\pi$
since Claim~\ref{lem:strongAMP} yields that the value of
$\vr a + \vr c + \vr t + \suma{\vr Z} + \vr{c''}$
is constantly equal to $A \cdot 4^B$ along $\pi$.

Finally, concerning the existence of the run $\pi$,
remark that, while the register updates
mentioned in Item~\ref{item:Init} and~\ref{item:Iter} can be applied irrespective of the values of $A$ and $B$,
the $Z$-zeroing calls to $\prog P$ described in Item~\ref{item:Call}
can be fulfilled if and only if $A$ is divisible
by a sufficiently large power of $2$.
More specifically, the run $\pi$ described in this proof exists
if and only if $4^{\successor{{\,F\,}}(B)-B}$ divides $A$.
\end{claimproof}

\subparagraph*{Bad runs.}
We prove that the bad runs of $\successor{{\,\prog P\,}}$
do not  $(Z \cup \{c''\})$-compute anything:

\begin{claim}\label{claim:badAmp}
Let $A,B \in \Npos$ be two positive integers.
Every bad run of $\successor{{\,\prog P\,}}$
starting in $\triple A B {\vr a} {\vr b''} {\vr c''} {\vr X \cup \{b'', c''\}}$
is not $(Z \cup \{c''\})$-zeroing.
\end{claim}

\begin{claimproof}
\renewcommand\qedsymbol{\textcolor{lipicsGray}{\ensuremath{\vartriangleleft}}}
Let $\pi$ be a run of $\successor{{\,\prog P\,}}$
starting in
$\triple A B {\vr a} {\vr b''} {\vr c''} {\vr X}$.
At the start of $\pi$ we have:
\[
\vr a = A, \quad
\vr{b''} = B,
\quad \vr{c''} = \vr a \cdot (4^{\vr{b''}} - 1), 
\]
and all other counters are 0.
In particular, this implies 
$
\vr b = \vr b' = \vr c = \vr c' = \vr t = 0 
$,
thus
\begin{align} \label{eq:invstep}
(\vr a + \vr c + \vr t + \suma{\vr Z}) \cdot 4^{\vr{b''}}
= \vr a + \vr c + \vr t + \suma{\vr Z} + \vr{c''}. 
\end{align}
We start by observing that
Claim~\ref{lem:strongAMP} implies that 
$\pi$ is $Z$-zeroing if and only if
Equation~\eqref{eq:invstep} holds after every step of $\pi$
and $b' = 0$ at the end of $\pi$:
The right-hand of Equation~\eqref{eq:invstep} side is preserved,
and the value of the left hand-side never increases,
thus if it ever decreases it remains smaller than the right-hand side
until the term of $\pi$, which in particular implies that 
the value of $\vr c''$ is not $0$.

Therefore we can immediately deduce that
if $\pi$ visits the iteration step
less than $B-1$ times,
then $b' > 0$ at the end of $\pi$,
thus $\pi$ is not $(Z \cup \{c''\})$-zeroing
by Equation~\eqref{eq:invstep}.
For the rest of this proof, let us suppose that $\pi$
visits the iteration step $B-1$ times.
Whenever $\pi$ goes through the initialisation step
or the iteration step,
it decrements $\vr b'$ by $1$, while gaining the opportunity
to increment the sum $\vr a + \vr c + \vr t + \suma{\vr Z}$, that we denote $\suma{\vr Z'}$
in order to maintain Equation~\eqref{eq:invstep}.
As we showed when we studied the good runs, in an ideal scenario
$\suma{\vr Z'}$ is quadrupled,
which compensates the decrement of $\vr b'$,
and Equation~\eqref{eq:invstep} still holds.
We now show that the occurrence of a single mistake at any point
results in $\suma{\vr Z'}$ not being quadrupled along a step:
We suppose that the run $\pi$ is bad,
we list all the possible errors it can commit,
and show that each one breaks Equation~\eqref{eq:invstep}:
\begin{itemize}
\item
If $\pi$ fails to maximally iterate one of the flat loops
at lines \ref{l:ini1n} or \ref{l:dec1n} then 
the sum $\suma{\vr Z'}$
is not quadrupled during the initialisation step:
Maximally iterating both loops (that is, iterating both of them $\vr a$ times)
increments $\suma{\vr Z'}$ by $3 \cdot \vr a$,
which exactly quadruples it since initially
the other variables occurring in $\vr S$ have value $0$.
However, since line~\ref{l:dec1n} increases $\suma{\vr Z'}$,
not maximally iterating it results in a smaller value.
Moreover, while line~\ref{l:ini1n} has no direct effect on $\suma{\vr Z'}$,
not maximally iterating it reduces the number of times
line~\ref{l:dec1n}
can be iterated, which in turn reduces the value of $\suma{\vr Z'}$.
\item
If $\pi$ fails to maximally iterate one of the flat loops
at lines \ref{l:startn}, \ref{l:endn} or \ref{l:endcn} then
the sum $\suma{\vr Z'}$ is not quadrupled
during the corresponding iteration step:
Maximally iterating the three loops
(that is, iterating lines \ref{l:startn} and \ref{l:endn} $\vr a$ times
and line \ref{l:endcn} $\vr c$ times)
increments $\suma{\vr Z'}$ by $3 \cdot (\vr a + \vr c)$,
which exactly quadruples it as long as the other variables occurring in $\suma{\vr Z'}$
had value $0$ to start with.
However, since lines \ref{l:endn} and \ref{l:endcn} increase $\suma{\vr Z'}$,
not maximally iterating them results in a smaller value.
Moreover, while line~\ref{l:startn} has no direct effect on $\suma{\vr Z'}$,
not maximally iterating it reduces the number of times line~\ref{l:endn}
can be iterated.
\item
If $\pi$ does a non $Z$-zeroing call to the program $\prog P$,
we differentiate two cases.
If this happens in the last iteration step we get
that $\pi$ is not $(Z \cup \{c''\})$-zeroing as it is not $Z$-zeroing.
If this happens in one of the previous iteration steps
then the sum $\suma{\vr Z'}$ is not quadrupled
in the \emph{next} iterations step:
as we just saw the iteration step increases $\vr S$ by at most
$3 \cdot (\vr a + \vr c)$,
which fails to quadruple it if there are nonzero counters in $\vr Z$.
\item
Finally,
let us consider the case where the first error committed by $\pi$
is failing to maximally iterate the flat loop at line~\ref{l:endbn}.
In this case, we show that the subsequent call to $\prog P$
is not $Z$-zeroing, which, as we have just shown,
implies that $\pi$ is not ($Z \cup {c''}$)-zeroing.
Since we assume that this is the first error committed by $\pi$,
up to this point, $\pi$ has behaved as a good run.
To analyse this situation,
let $\pi'$ be the run that behaves as $\pi$ up to this point
but then maximally iterates the flat loop at line~\ref{l:endbn}.
By Lemma~\ref{claim:goodAmp}, we know that $\pi'$ enters
the call to $\prog P$ with a counter valuation matching
$\triple {A\cdot 4^{i+2-F^{i}(1)}} {F^{i}(1)} {\vr a} {\vr b'} {\vr c'} {\vr X}$
for some $0 \leq i \leq B-1$. 
In particular, the following equation holds for $\pi'$:
\[
{(\vr a + \vr c + \vr t + \suma{\vr Z}) \cdot 4^{\vr{b'}}} =
{A \cdot 4^{i+2}} =
\vr a + \vr c + \vr t + \suma{\vr Z} + \vr c' 
\]
However, since $\pi$ did \emph{not} maximally iterate line~\ref{l:endbn},
$\vr b'$ will be smaller in $\pi$ compared to $\pi'$
(and $\vr b$ will be larger - but this has no impact on the following argument
since $\vr b \notin \vr Z$).
Consequently, $\pi$ will call the program $\prog P$ 
with a counter valuation satisfying:
\[
(\vr a + \vr c + \vr t + \suma{\vr Z}) \cdot 4^{\vr{b'}} <
\vr a + \vr c + \vr t + \suma{\vr Z} + \vr c'.
\]
Since $\prog P$ is a strong amplifier,
this equation still holds at the exit of $\prog P$.
In particular, this implies that $\vr c'$ is not $0$,
thus the call to $\prog P$ is not $Z$-zeroing.\qedhere
\end{itemize}
\end{claimproof}
\end{toappendix}

\section{Amplifiers defined by counter programs with a pushdown stack}\label{sec:pampl}
The goal of this section is to prove Lemma~\ref{lem:amppushdown}.
\medskip
\\
Our construction is based on the amplifiers from Section \ref{sec:ampl}.
The main idea is to ``delegate'' some counters to the stack.
The stack alphabet consists exactly of those counters which are delegated, and the value of each 
delegated counter $\vr x$ corresponds to the number
of occurrences of the symbol  $\vr x$ on the stack.
Therefore, ``delegated counters'' can be understood as a synonym of ``stack symbols'' in the sequel.
This idea motivates the following definition.
Let $\vr S$ and $\vr X$ be two disjoint sets of delegated,
respectively non-delegated, counters.
We define the function 
\[
\vals{\vr X, \vr S} : \N^{\vr X} \times \vr S^* \to \N^{\vr X\cup \vr S}
\]
that maps a configuration, i.e., a valuation $v$ of the non-delegated counters of $\vr X$
together with a stack content $s\in \vr S^*$, 
to a valuation of all the counters from $\vr X\cup \vr S$,
as follows: 
$\vals{\vr X, \vr S}(v, s) = v'$, where
$v'({\vr x}) = v(\vr x)$ for $\vr x \in \vr X$, and
$v'({\vr x})$ is the number of occurrences of $\vr x$ in $s$ for $\vr x\in \vr S$. 

Using this definition, we establish a notion of \emph{simulation}
between programs with or without stack.
Given an $F$-amplifier $\prog P$ with set of counters $\vr X \cup \vr S$,
we say that a counter program with a stack $\prog Q$
\emph{simulates} $\prog P$
if it satisfies the two following conditions:
\begin{itemize}
\item 
For every $A,B \in \Npos$ such that $A$ is divisible by $4^{(F(B)-B)}$
there exists a run of $\prog Q$ between two configurations $x$ and $y$
satisfying $\vals{\vr X,\vr S}(x) = \triple {A} {B} {\vr a} {\vr b'} {\vr c'} {\vr X}$
and $\vals{\vr X,\vr S}(y) =
\triple {A\cdot 4^{(B-F(B))}} {F(B)} {\vr a} {\vr b} {\vr c} {\vr X}$.
\item 
For every run of $\prog Q$ between two configurations $x$ and $y$
there exists a run of $\prog P$ between $\vals{\vr X,\vr S}(x)$
and $\vals{\vr X,\vr S}(y)$;
\end{itemize}
We say that such a program with stack $\prog Q$ is an $F$-amplifier.

The rest of this section is devoted to the proof of Lemma~\ref{lem:amppushdown}.
We rely on the constructions of Section~\ref{sec:ampl},
and similarly proceed in two steps.
First, we transform the $\F 1$-amplifier $\prog P_1$
into a program with stack $\prog Q_1$
that simulates $\prog P_1$ with four counters,
as it delegates the two other counters to the stack (Lemma~\ref{lem:Q1}).
Next, we adapt the constructions used to lift $F$-amplifiers
into $\successor{{\,F\,}}$-amplifier.
This time, we will have \emph{two} constructions that can be applied alternatively:
the first introduces one counter and one delegated counter,
and the second introduces two delegated counters (Lemma~\ref{lemma:progQ}).
Therefore, for every $d \in \mathbb{N}$,
starting with the program $\prog Q_1$
and applying alternatively our two lifting constructions
yields a $\F {d}$-amplifier with $\lfloor \frac d 2 \rfloor + 4$ counters
(as the other counters are delegated to the stack),
which proves Lemma~\ref{lem:amppushdown}.

\subsection{Construction of the \texorpdfstring{\boldmath $\F 1$}{F\_\{\#1\} 1}-amplifier \texorpdfstring{\boldmath $\prog Q_1$}{\#1 Q\_1}}
Consider the following program with 4 counters $\vr X = \set{\vr a, \vr b, \vr c, \vr t}$
and the stack alphabet $\vr S = \set{\vr b', \vr c'}$,
which is obtained from the program $\prog P_1$ defined in Section \ref{sec:ampl} by replacing each decrement
on $\vr b'$ and $\vr c'$ by the corresponding pop operation:

\medskip

\PROG{0.9}{Program $\prog Q_1$}{prog:Q1}{
\LLoop{$\vr a \tran t$}
\label{l2:startini}
\LLoop{$\vr t \tran a$  \quad \pop{\vr c'c'c'} \quad \label{l2:dec1}
\add{\vr c}{3}}\label{l2:pop0}
\State \pop{\vr b'}  \quad \add{\vr b} {1} \label{l2:pop1}
\label{l2:endini} 
\Loop 
\label{l2:startit} 
\LLoop{$\vr a \tran t$ \quad \pop{\vr c'} \quad \add{\vr t}{1}}
\LLoop{$\vr c \tran a$ \quad \pop{\vr c'} \quad \add{\vr a}{1}\label{l2:dec2}}\label{l2:pop2}
\LLoop{$\vr a \tran c$ \quad \pop{\vr c'} \quad \add{\vr c}{1}}
\LLoop{$\vr t \overset{8}{\tran} c$  \quad \pop{\vr c'c'c'c'c'c'c'c'} \quad \add{\vr a}{1} \quad \add{\vr c}{7}\label{l2:dec3}}
\State \pop{\vr b'} \quad \add{\vr b} {2} \label{l2:pop3}
\EndLoop\label{l:endit} 
}

\medskip

\begin{lemma}\label{lem:Q1}
The program $\prog Q_1$ simulates the $\F 1$-amplifier
$(\prog P_1,\vr X,(a,b,c),(a,b',c'),\vr t,\{c''\})$.
\end{lemma}

\begin{proof}
Let $\vals{}$ denote the function $\vals{\{\vr a,b,c,t\}, \{\vr b',c'\}}$
that transforms configurations of $\prog Q_1$ into configurations of $\prog P_1$.
The program $\prog Q_1$ is a constrained version of $\prog P_1$:
every line is identical with the added restriction that lines~\ref{l2:pop0},~\ref{l2:pop1},~\ref{l2:pop2} and~\ref{l2:pop3}
can only be fired if the appropriate symbol is at the top of the stack.
Therefore, we immediately get the second condition required
for $\prog Q_1$ to simulate $\prog P_1$:
for every run $\pi$ of $\prog Q_1$ between two configurations $x$ and $y$,
the run $\pi'$ of $\prog P$
that starts in $\vals{}(x)$ and uses the same lines as $\pi$
is a valid run of $\prog P$ that ends in $\vals{}(y)$.

To conclude, we show that we can transform
the $\{\vr b', c'\}$-zeroing runs of $\prog P_1$
(thus in particular the runs that witness the $\F 1$-amplifier behaviour)
into runs of $\prog Q_1$.
To do so we rely on the fact 
that the counters $\vr b'$ and $\vr c'$
are only decreasing along the runs of $\prog P_1$.
Formally, given a $\{\vr b', c'\}$-zeroing run $\pi$ of $\prog P_1$
between two configurations $x$ and $y$,
let $u_\pi \in \{\vr b', c'\}^*$
be the word listing, in order, the occurrences
of the decrements of $\vr b'$ and $\vr c'$ along $\pi$.
We define a configuration $x'$ of $\prog Q_1$ as follows:
the counters $\vr a, b, c, t$ match the content they have
in the starting configuration $x$ of $\pi$,
and the stack content is the reverse of the word $u_\pi$ 
so that the first letter of $u_{\pi}$ is at the top of the stack.
This definition guarantees that:
\begin{itemize}
\item We have $\vals{}(x') = x$.
Indeed, since $\pi$ is $\{\vr b', c'\}$-zeroing,
the value of the counters $\vr b'$ and $\vr c'$ in the initial configuration $x$
is equal to the number of times these counters are decremented;
\item There exists a run $\pi'$ of $\prog Q_1$ that starts from $x'$
and follows the same lines as $\pi$:
whenever a popping instruction appears
the adequate symbol will be at the top of the stack.
As the lines of $\prog P_1$ and $\prog Q_1$ are analogous,
the configuration $y'$ reached by $\pi'$ satisfies
$\vals{}(y') = y$.\qedhere
\end{itemize}
\end{proof}

\subsection{Construction of the \texorpdfstring{\boldmath $\successor{{\,F\,}}$}{F \textasciitilde}-amplifiers \texorpdfstring{\boldmath $\successor{{\, \prog Q\,}}$}{\#1 Q \textasciitilde} and \texorpdfstring{\boldmath $\overline{{\, \prog Q\,}}$}{overline\{\{\#1 Q\}\}} from an \texorpdfstring{\boldmath $F$}{F}-amplifier \texorpdfstring{\boldmath $\prog Q$}{\#1 Q}}
In Section~\ref{sec:ampl},
we showed how to lift an $F$-amplifier $\prog P$ into an
$\successor{{\,F\,}}$-amplifier $\successor{{\,\prog P\,}}$.
We now show two different manners of adapting the construction 
of $\successor{{\,\prog P\,}}$ in order to
lift a program with stack $\prog Q$ simulating $\prog P$
into a program with stack simulating $\successor{{\,\prog P\,}}$.

\begin{restatable}{lemma}{progQ}\label{lemma:progQ}\sloppy
Let $\prog Q$ be a program simulating a strong $F$-amplifier
$(\prog P,\vr X,(a,b,c),(a,b',c'),t,Z)$
without delegating the counters $\vr a, b, c$ and $\vr t$.
\begin{itemize}
\item 
If $\prog Q$ delegates $\vr b'$ but not $\vr c'$,
then $\successor{{\,\prog Q\,}}$ simulates
$(\successor{{\,\prog P\,}},\vr X \cup \{\vr{b''}, \vr{c''}\},(a,b,c),(a,b'',c''),Z \cup \{c''\})$
while delegating two input counters $\vr b''$ and $\vr c''$
in addition to the counters delegated by $\prog Q$.
\item
If $\prog Q$ delegates $\vr b'$ and $\vr c'$,
then $\overline{{\,\prog Q\,}}$ simulates $(\successor{{\,\prog P\,}},\vr X \cup \{\vr{b''}, \vr{c''}\},(a,b,c),(a,b'',c''),Z \cup \{c''\})$
while delegating only one input counter $\vr b''$
in addition to the counters delegated by $\prog Q$.
\end{itemize}
\end{restatable}
\hspace{-1em}
\begin{minipage}[t]{0.49\textwidth}
\strut\vspace*{-\baselineskip}

\PROG{0.95}{Program $\successor{{\,\prog Q\,}}$}{prog:Q}{
\LLoop{$\vr a \tran t$}\label{l:startinit} 
\LLoop{$\vr t \tran a$ \quad \pop{\vr{c''}\vr{c''}\vr{c''}} \quad \add{\vr c}{3}}\State \pop{\vr{b''}} \quad \add{\vr b} {1}\label{l:endinit}    \Loop\label{l:startmain} 
\LLoop{$\vr a \tran t$}\label{l:it1}
\LLoop{$\vr t \tran a$ \ \pop{\vr{c''}\vr{c''}\vr{c''}} \ \add{\vr a}{3}}\label{l:end3loop}\label{l:it2}
\LLoop{$\vr c \tran c'$ \ \pop{\vr{c''}\vr{c''}\vr{c''}} \ \add{\vr c'}{3}}\label{l:begc}\label{l:endc}\label{l:it3}
\LLoop{\dec{\vr b} \quad \push{\vr b'}}\label{l:begb}\label{l:endb}\label{l:it4}
\State $\prog Q$ 
\State \pop{\vr{b''}} 
\EndLoop\label{l:endmain} 
}
\end{minipage}\begin{minipage}[t]{0.49\textwidth}
\strut\vspace*{-\baselineskip}

\PROG{0.95}{Program $\overline{{\,\prog Q\,}}$}{prog:Qbis}{
\LLoop{$\vr a \tran t$}
\LLoop{$\vr t \tran a$  \quad \sub{\vr{c''}}{3} \quad \add{\vr c}{3}}
\State \pop{\vr{b''}}  \quad \add{\vr b} {1}   \Loop 
\LLoop{$\vr a \tran t$}
\LLoop{$\vr t \tran a$\quad \sub{\vr{c''}}{3} \quad \add{\vr a}{3}}
\Loop \label{l:bego}
\LLoop{\dec{\vr b} \quad \push{\vr b'} }
\State  \dec{\vr c} \quad \sub{\vr{c''}}{3} \quad \push{\vr c'} 
\LLoop{\dec{\vr b} \quad \push{\vr b'}}
\State \push{\vr c'}  
\LLoop{\dec{\vr b} \quad \push{\vr b'}}
\State \push{\vr c'}  
\LLoop{\dec{\vr b} \quad \push{\vr b'}}
\State \push{\vr c'}  
\LLoop{\dec{\vr b} \quad \push{\vr b'}}\label{l:endo}
\EndLoop 
\State $\prog Q$ 
\State \pop{\vr{b''}} 
\EndLoop
}
\end{minipage}

\medskip

\begin{toappendix}
\section{Proof of Lemma~\ref{lemma:progQ}}\label{appendix:progQ}

Let $\vr X$ and $\vr S$ denote the set of counters of $\prog Q$, respectively its stack alphabet.
Let $\successor{{\,\vals{}\,}}$ denote the function $\vals{\vr X, \vr S \cup \{\vr b'',c''\}}$
that transforms configurations of $\successor{{\,\prog Q\,}}$ into configurations of $\successor{{\,\prog P\,}}$.
Similarly, let $\overline{{\,\vals{}\,}}$ denote the function $\vals{\vr X \cup \{\vr c''\}, \vr S \cup \{\vr b''\}}$
that transforms configurations of $\overline{{\,\prog Q\,}}$ into configurations of $\successor{{\,\prog P\,}}$.
The proof is done in three steps.
First, we show that we can easily translate
the runs of $\successor{{\,\prog Q\,}}$ and $\overline{{\,\prog Q\,}}$
into runs of $\successor{{\,\prog P\,}}$ with matching source and target.
The harder part of the proof is to show the reciprocal statement:
we consider the good runs of $\successor{{\,\prog P\,}}$
described in Lemma~\ref{claim:goodAmp} 
and we show how to translate them,
first into runs of $\successor{{\,\prog Q\,}}$,
and finally into runs of $\overline{{\,\prog Q\,}}$.

\subparagraph*{Transforming runs of \texorpdfstring{\boldmath $\successor{{\,\prog Q\,}}$}{\#1 Q \textasciitilde} and \texorpdfstring{\boldmath $\overline{{\,\prog Q\,}}$}{overline\{\{\#1 Q\}\}} into runs of \texorpdfstring{\boldmath $\successor{{\,\prog P\,}}$}{\#1 P \textasciitilde}.}
The program $\successor{{\,\prog Q\,}}$ is a constrained version of $\successor{{\,\prog P\,}}$:
every line is identical except for the lines
with a popping instruction instead of a decrement, which are more restrictive
since the correct symbol needs to be at the top of the stack.
As a consequence, for every run $\successor{{\,\pi\,}}$
of $\successor{{\,\prog Q\,}}$
between two configurations $x$ and $y$
we immediately get a run $\pi$
of $\successor{{\,\prog P\,}}$
between $\successor{{\,\vals{}\,}}(x)$ and $\successor{{\,\vals{}\,}}(y)$
which uses the same lines as $\successor{{\,\pi\,}}$.

Now given a run of $\overline{{\,\pi\,}}$ of $\overline{{\,\prog Q\,}}$ between two configurations $x$ and $y$,
translating $\overline{{\,\pi\,}}$ into a run of $\successor{{\,\prog P\,}}$ is not as direct
since the lines \ref{l:bego}--\ref{l:endo} of $\overline{{\,\prog Q\,}}$
are not exactly analogous to the lines \ref{l:endcn}--\ref{l:endbn} of $\successor{{\,\prog P\,}}$.
However, as we explained in the paper,
using some local reshuffling
$\successor{{\,\prog P\,}}$ can reproduce any counter update corresponding to the lines
\ref{l:bego}--\ref{l:endo} of $\overline{{\,\prog Q\,}}$.
This allows us to transform the run $\overline{{\,\pi\,}}$
into a run of $\successor{{\,\prog P\,}}$
between $\overline{{\,\vals{}\,}}(x)$ and $\overline{{\,\vals{}\,}}(y)$.

\subparagraph*{Transforming runs of \texorpdfstring{\boldmath $\successor{{\,\prog P\,}}$}{\#1 P \textasciitilde} into runs of \texorpdfstring{\boldmath $\successor{{\,\prog Q\,}}$}{\#1 Q \textasciitilde}.}
Let us suppose that $\prog Q$ delegates the counter $\vr b'$ but not $\vr c'$,
and let $\pi$ be a \emph{good} run of $\successor{{\,\prog P\,}}$
as described in the proof of Lemma~\ref{claim:goodAmp}.
We denote by $x$ and $y$ the starting and ending configuration of $\pi$.
In order to transfer $\pi$ to $\successor{{\,\prog Q\,}}$,
we begin by creating an appropriate initial configuration $x'$
as in the proof of  Lemma~\ref{lem:Q1}.
Formally, let $u_\pi \in \{\vr b', c'\}^*$
be the word listing, in order, the occurrences
of the decrements of $\vr b''$ and $\vr c''$ along $\pi$.
We define the configuration $x'$ of $\successor{{\,\prog Q\,}}$ by setting
the values of the counters of $\vr X$ to the values they have in 
the starting configuration $x$ of $\pi$,
and setting the stack content to the reverse of the word $u_\pi$ 
(so that the first letter of $u_{\pi}$ is at the top of the stack).
This definition guarantees that $\vals{}(x') = x$.
To conclude, we need to argue that $\successor{{\,\prog Q\,}}$
can simulate $\pi$ starting from $x'$.
First, remark that the initialisation step is easily simulated since
the definition of the initial stack content guarantees that the appropriate symbol
is at the top of the stack whenever needed.
We now explain, step by step, how $\successor{{\,\prog Q\,}}$ simulates the iteration steps of $\pi$.
First, thanks to the definition of the initial stack content the loops on lines \ref{l:it1}--\ref{l:it3} can be iterated as in $\pi$.
Then, we also iterate line \ref{l:it4} as in $\pi$.
Remark that this disrupts the stack by adding a sequence of $\vr b'$ on top of it.
Next comes the call to $\prog P$, and since  $\pi$ is a good run
we know that this call is \emph{correct},
in the sense that it starts in
$\triple A B {\vr a} {\vr b'} {\vr c'} {\vr X}$ (Equation~\eqref{eq:y})
and ends in
$\triple {A\cdot 4^{B-\successor{{\,F\,}}(B)}} {\successor{{\,F\,}}(B)} {\vr a} {\vr b} {\vr c} {\vr X}$ (Equation~\eqref{eq:x}) for some $A,B \in \Npos$.
Therefore in $\successor{{\,\prog Q\,}}$ we can simulate this call to $\prog P$ by a call to $\prog Q$,
and since the value of $\vr b'$ is $0$ in the ending configuration
this implies that the call to $\prog Q$ will automatically pop all of the $\vr b'$ that were added on the stack.
Therefore we are back with a stack content that matches a prefix of our initial stack content,
and we can conclude the simulation of the iteration step by popping a single $\vr b''$ from the stack.

\subparagraph*{Transforming runs of \texorpdfstring{\boldmath $\successor{{\,\prog P\,}}$}{\#1 P \textasciitilde} into runs of \texorpdfstring{\boldmath $\overline{{\,\prog Q\,}}$}{overline\{\{\#1 Q\}\}}.}
Let us suppose that $\prog Q$ delegates both $\vr b'$ and $\vr c'$,
and let $\pi$ be a \emph{good} run of $\successor{{\,\prog P\,}}$,
as described in the proof of Lemma~\ref{claim:goodAmp}.
We denote by $x$ and $y$ the starting and ending configuration of $\pi$.
We show how to construct a run $\overline{\pi}$ of $\overline{{\,\prog Q\,}}$ that simulates $\pi$.
First, remark that we have a single possibility
for the starting configuration of $\overline{\pi}$:
In the starting configuration of $\pi$
only the values of $\vr a, \vr b''$ and $\vr c''$ are nonzero (Equation~\eqref{eq:w}),
and $\overline{{\,\prog Q\,}}$ only delegates $\vr b''$ among these three counters.
Therefore the initial stack content will just be a sequence of  $\vr b''$ of the appropriate length.
Then, simulating the initialisation step of $\pi$ is easy:
one $\vr b''$ is popped from the stack and the other counters
are updated as in $\pi$.

To conclude, we show how to simulate
the iteration steps visited by $\pi$.
Let $\pi_1 \pi_2 \pi_3$ be a subrun of $\pi$ corresponding to an iteration step
of $\successor{{\,\prog P\,}}$,
where $\pi_2$ stands for the call to the program $\prog P$.
As we showed in the proof of Lemma~\ref{claim:goodAmp},
every call to $\prog P$ along $\pi$ is \emph{correct},
in the sense that it starts in some configuration
$\triple A B {\vr a} {\vr b'} {\vr c'} {\vr X}$ (Equation~\eqref{eq:y})
and ends in
$\triple {A\cdot 4^{B-\successor{{\,F\,}}(B)}} {\successor{{\,F\,}}(B)} {\vr a} {\vr b} {\vr c} {\vr X}$ for some $A,B \in \Npos$ (Equation~\eqref{eq:x}).
As a consequence, since $\prog Q$ simulates $\prog P$,
there exists a run $\overline{\pi}_2$
of $\prog Q$ that simulates $\pi_2$,
but this run requires a starting stack content
corresponding to some specific shuffle $u_{\pi_2}$
of the word $(\vr b')^B(\vr c')^{A \cdot (4^B-1)}$.
Fortunately, as we explained in the paper,
the lines \ref{l:bego}--\ref{l:endo} of $\overline{{\,\prog Q\,}}$
allow to push any shuffle of $\vr b'$ and $\vr c'$ on the stack.
In particular, there exists a subrun $\overline{\pi}_1$ of $\overline{{\,\prog Q\,}}$
that simulates $\pi_1$ and pushes the word $u_{\pi_2}$ on the stack.
As a consequence, $\overline{\pi}_1\overline{\pi}_2$
simulates truthfully the subrun $\pi_1\pi_2$ with no impact on the stack:
$\overline{\pi}_1$ pushes $u_{\pi_2}$, which is then popped by $\overline{\pi}_2$.
Therefore, we can simulate the iteration step $\pi_1 \pi_2 \pi_3$
by starting with $\overline{\pi}_1\overline{\pi}_2$,
and then adding $\overline{\pi}_3$ which pops a single $\vr b$ from the stack
to simulate the decrement of $\vr b$ occurring in $\pi_3$.
\end{toappendix}
~\\
The proof of Lemma~\ref{lemma:progQ}
can be found in Appendix~\ref{appendix:progQ}.
To convey the intuition behind it we
analyse the differences between the two programs.
The main difference concerns the counters delegated to the stack:
If $\prog Q$ delegates only $\vr b'$, then the starting configurations for the calls to $\prog Q$ are easy to setup
as the stack simply contains a sequence of $\vr b'$.
Therefore $\prog Q$ can be lifted via $\successor{{\,\prog Q\,}}$
which delegates both $\vr b''$ and $\vr c''$.
However, if $\prog Q$ delegates both $\vr b'$ and $\vr c'$, then
the starting configurations required for the calls to $\prog Q$ are more complex:
the stack needs to contain the symbols $\vr b'$ and $\vr c'$ in a specific order.
This prevents us from delegating both $\vr b''$ and $\vr c''$ to the stack,
thus we need to lift $\prog Q$ via $\overline{{\,\prog Q\,}}$
which delegates only $\vr b''$
and keeps  $\vr c''$ as a standard counter.
A second difference between $\successor{{\,\prog Q\,}}$ and $\overline{{\,\prog Q\,}}$ concerns
the loops updating $\vr b'$ and $\vr c'$ in the iteration step.
To understand what is happening here, let us have a look at what happens when we replace 
the push and pop instructions by increments and decrements:

\hspace{-1em}
\begin{minipage}[t]{0.5\textwidth}
\strut\vspace*{-\baselineskip}

\PROGnoname{0.95}{prog:2loops}{
\LLoop{$\vr c \tran c'$ \quad \sub{\vr{c''}}{3} \quad \add{\vr c'}{3}}\label{l:c'1}
\LLoop{$\vr b \tran b'$}
}
\end{minipage}
\begin{minipage}[t]{0.47\textwidth}
\strut\vspace*{-\baselineskip}

\PROGnoname{0.95}{prog:b}{
\Loop 
\LLoop{$\vr b \tran b'$}
\State{\dec{\vr c} \quad \inc{\vr c'} \quad \sub{\vr{c''}}{3}}\label{l:c''1}
\LLoop{$\vr b \tran b'$}
\State{\inc{\vr c'}}
\LLoop{$\vr b \tran b'$}
\State{\inc{\vr c'}}
\LLoop{$\vr b \tran b'$}
\State{\inc{\vr c'}}
\LLoop{$\vr b \tran b'$}
\EndLoop 
}
\end{minipage}
\medskip
\\
While these two sequences of instructions are different, we can remark that their global effect
is identical, in the sense that every counter update realisable by the left one is also realisable by the right one,
and reciprocally.
However, if $\vr b'$ and $\vr c'$ are delegated
to the stack then the sequence of instruction on the right is more powerful,
as it performs the same number of increments of
$\vr b'$ and $\vr c'$, but \emph{in any order}, which allows to create many different stack contents.
This is required so that $\overline{{\,\prog Q\,}}$
can construct the stack contents needed to call $\prog Q$.

\bibliography{p035-Czerwinski}

\end{document}